\documentclass[11pt]{article}
\usepackage[utf8]{inputenc}

\usepackage[utf8]{inputenc}
\usepackage[utf8]{inputenc}
\usepackage[T1]{fontenc}
\usepackage[sc]{mathpazo}
\usepackage{fullpage}
\usepackage{amsthm,bm}
\usepackage{mathtools}
\usepackage{amsmath}
\usepackage{mathrsfs}
\usepackage{amssymb}
\usepackage{algorithm}
\usepackage[noend]{algpseudocode}
\usepackage{blkarray}
\usepackage{dsfont}
\usepackage{enumitem}
\usepackage{marginnote}
\usepackage{multirow}
\usepackage{url}
\usepackage[pdftex, colorlinks,
  linkcolor=blue,
  citecolor=purple,
  filecolor=blue,urlcolor=blue]%
{hyperref}
  
\usepackage{cleveref}

\newtheorem{theorem}{Theorem}[section]
\newtheorem{lemma}[theorem]{Lemma}
\newtheorem{definition}[theorem]{Definition}
\newtheorem{proposition}[theorem]{Proposition}

\newtheorem{corollary}[theorem]{Corollary}

\newtheorem{example}[theorem]{Example}

\let\al=\alpha
\let\vf=\varphi
\let\sg=\sigma
\let\dl=\delta

\let\Dl=\Delta
\let\om=\omega
\def\cA{\mathcal A}

\def\cC{\mathcal C}
\def\cD{\mathcal D}

\def\cP{\mathcal P}

\def\cX{\mathcal X}
\def\cY{\mathcal Y}

\def\zA{\mathbb A}
\def\zB{\mathbb B}
\def\zC{\mathbb C}

\def\zF{\mathbb F}
\def\zZ{\mathbb Z}
\def\zR{\mathbb R}
\def\zN{\mathbb N}
\def\GF{\mathsf{GF}}
\def\ba{\mathbf{a}}
\def\bb{\mathbf{b}}
\def\bc{\mathbf{c}}
\def\bd{\mathbf{d}}

\let\ov=\overline
\let\sse=\subseteq
\let\tm=\times
\def\vc#1#2{#1_1,\dots,#1_{#2}}
\def\cl#1#2{\arraycolsep0pt
\left(\begin{array}{c} #1\\ #2 \end{array}\right)}
\def\zd{,\dots,}

\def\ang#1{\langle #1\rangle}

\def\bs{\mathbf{s}}
\def\bx{\mathbf{x}}

\def\pr{\mathrm{pr}}

\newcommand{\Field}{\mathbb{F}}
\newcommand{\Real}{\mathbb{R}}

\newcommand{\Complex}{\mathbb{C}}

\newcommand{\CSP}{\textsc{CSP}}
\newcommand{\MCSP}{\textsc{CSP}}
\newcommand{\IMP}{\textsc{IMP}}
\newcommand{\xIMP}{\mbox{$\chi$\textsc{IMP}}}

\newcommand{\Pol}{\textsf{Pol}}

\newcommand{\Sos}{\textsf{SOS}}

\newcommand{\Sol}{\textsf{Sol}}
\newcommand{\I}{\emph{\texttt{I}}}
\newcommand{\J}{\emph{\texttt{J}}}
\def\zQ{\mathbb Q}
\let\Gm=\Gamma
\newcommand{\mc}[1]{\mathcal{#1}}
\newcommand{\mb}[1]{\mathbf{#1}}
\newcommand{\Variety}[1]{{\textbf{V}}\left( #1 \right)}
\newcommand{\Ideal}[1]{\left\langle #1 \right\rangle}

\newcommand{\lex}{\textsf{lex }}

\newcommand{\grlex}{\textsf{grlex }}
\newcommand{\grlexns}{\textsf{grlex}}

\newcommand{\coNPc}{\text{\textbf{coNP}-complete}}

\newcommand{\multideg}{\textnormal{multideg}}
\newcommand{\LM}{\textnormal{LM}}
\newcommand{\LT}{\textnormal{LT}}
\newcommand{\LC}{\textnormal{LC}}
\newcommand{\LCM}{\textnormal{lcm}}
\newcommand{\GB}{\text{Gr\"{o}bner} Basis}
\newcommand{\GBs}{\text{Gr\"{o}bner} Bases}

\usepackage{xcolor}

\title{The Ideal Membership Problem and Abelian Groups\footnote{An extended abstract of this work appeared in the \emph{Proceedings of the 39th International Symposium on Theoretical Aspects of Computer Science (STACS 2022) \cite{stacs/BulatovR22}}.}}
 \author{Andrei A. Bulatov\thanks{{abulatov@sfu.ca}. Department of Computing Science, Simon Fraser University, Burnaby, BC, Canada. Research supported by an NSERC Discovery Grant.} \and Akbar Rafiey\thanks{{akbar.rafiey@gmail.com, ar9530@nyu.edu}. Department of Computer Science and Engineering, Tandon
School of Engineering, New York University, NY,  USA}}
\date{}

\begin{document}

\maketitle
\begin{abstract}
    Given polynomials $f_0,\vc fk$ the Ideal Membership Problem, IMP for short, asks if $f_0$ belongs to the ideal generated by $\vc fk$. In the search version of this problem the task is to find a proof of this fact. The IMP is a well-known fundamental problem with numerous applications. For instance, it underlies many proof systems based on polynomials such as Nullstellensatz, Polynomial Calculus, and Sum-of-Squares. Although the IMP is in general intractable, in many important cases it can be efficiently solved.
    
    Mastrolilli [SODA'19] initiated a systematic study of IMPs for ideals arising from Constraint Satisfaction Problems (CSPs), parameterized by constraint languages, denoted $\IMP(\Gm)$. The ultimate goal of this line of research is to classify all such IMPs accordingly to their complexity. Mastrolilli achieved this goal for IMPs arising from $\CSP(\Gm)$ where $\Gamma$ is a Boolean constraint language, while Bulatov and Rafiey [STOC'22] advanced these results to several cases of CSPs over finite domains. In this paper we consider IMPs arising from CSPs over `affine' constraint languages, in which constraints are subgroups (or their cosets) of direct products of Abelian groups. This kind of CSPs include systems of linear equations and are considered one of the most important types of tractable CSPs. Some special cases of the problem have been considered before by Bharathi and Mastrolilli [MFCS'21] for linear equations modulo 2, and by Bulatov and Rafiey [STOC'22] for systems of linear equations over $\GF(p)$, $p$ prime. Here we prove that if $\Gm$ is an affine constraint language then $\IMP(\Gm)$ is solvable in polynomial time assuming the input polynomial has bounded degree. 
\end{abstract}




\newpage

\section{Introduction}

\paragraph{The Ideal Membership Problem.}
Representing combinatorial problems by polynomials and then using algebraic techniques to approach them is one of the standard methods in algorithms and complexity. The Ideal Membership Problem (IMP for short) is an important algebraic framework that has been instrumental in such an approach. The IMP  underlies many proof systems based on polynomials such as Nullstellensatz, Polynomial Calculus, and Sum-of-Squares, and therefore plays an important role in such areas as proof complexity and approximation.

Let $\zF$ be a field and $\zF[\vc xn]$ be the ring of polynomials over $\zF$. Given polynomials $f_0$, $\vc fk\in\zF[\vc xn]$ the IMP asks if $f_0$ belongs to the ideal $\ang{\vc fk}$ of $\zF[\vc xn]$ generated by $\vc fk$. This fact is usually demonstrated by presenting a \emph{proof}, that is, a collection of polynomials $\vc hk$ such that the following polynomial identity holds $f_0=h_1f_1+\dots+h_kf_k$. Many applications require the ability to produce such a proof. We refer to this as \emph{finding a membership proof} problem. Note that by the Hilbert Basis Theorem any ideal of $\zF[\vc xn]$ can be represented by a finite set of generators meaning that the above formulation of the problem covers all possible ideals of $\zF[\vc xn]$ (assuming also that the elements of $\zF$ are finitely presented). 

The general IMP is a difficult problem and it is not even obvious whether or not it is decidable. The decidability was established in \cite{hermann1926frage,richman1974constructive,seidenberg1974constructions}. Then Mayr and Meyer~\cite{mayr1982complexity} were the first to study the complexity of the IMP. They proved an exponential space lower bound for the membership problem for ideals generated by polynomials with integer and rational coefficients. Mayer~\cite{Mayr89} went on establishing an exponential space upper bound for the $\IMP$ for ideals over $\zQ$, thus proving that such IMPs are \textbf{EXPSPACE}-complete. The source of hardness here is that a proof that $f_0\in\ang{\vc fk}$ may require polynomials of exponential degree. (There is also the issue of exponentially long coefficients that we will mention later.)

\paragraph{Combinatorial ideals.}
To illustrate the connection of the IMP to combinatorial problems we consider the following simple example. We claim that the graph in Fig.~\ref{fig:example1} is 2-colorable if and only if polynomials
\[
x(1-x), y(1-y), z(1-z), x+y-1, x+z-1,y+z-1
\]
have a common zero. Indeed, denoting the two possible colors 0 and 1, the first three polynomials guarantee that the only zeroes this collection of polynomials can have are such that $x,y,z\in\{0,1\}$. Then the last three polynomials make sure that in every common zero the values of $x,y,z$ are pairwise different, and so correspond to a proper coloring of the graph. Of course, the graph in the picture is not 2-colorable, and by the Weak Nullstellensatz this is so if and only if the constant polynomial 1 belongs to the ideal generated by the polynomials above. A proof of that can be easily found 
\[
1=(-4)\left[x(x-1)\right]+\left(2x-1\right)\left([x+y-1]-[y+z-1]+[x+z-1]\right).
\]

\begin{figure}[ht]
\centerline{\includegraphics[totalheight=3cm,keepaspectratio]{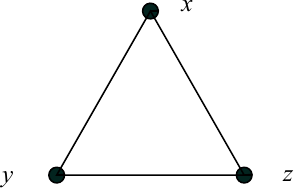}}
\caption{Graph 2-colorability}
\label{fig:example1}
\end{figure}

The example above exploits the connection between polynomial ideals and sets of zeroes of polynomials, also known as \emph{affine varieties}. While this connection does not necessarily hold in the general case, as Hilbert's Nullstellensatz requires certain additional properties of ideals, it works for so called \emph{combinatorial ideals} that arise from the majority of combinatorial problems similar to the example above. The varieties corresponding to combinatorial ideals are finite, and the ideals themselves are zero-dimensional and radical. These properties make the IMP significantly easier, in particular, it can be solved in single exponential time~\cite{DickensteinFGS91}. Also, Hilbert's Strong Nullstellensatz holds in this case, which means that if the IMP is restricted to radical ideals, it is equivalent to (the negation of) the question: given $f_0,\vc fk$ does there exist a zero of $\vc fk$ that is not a zero of $f_0$.

The special case of the $\IMP$ with $f_0=1$ has been studied for combinatorial problems in the context of lower bounds on Polynomial Calculus and Nullstellensatz proofs, see e.g. \cite{BeameIKPP94,BussP96,Grigoriev98}. A broader approach of using polynomials to represent finite-domain constraints has been explored in  \cite{CleggEI96,JeffersonJGD13}. Clegg et al., \cite{CleggEI96}, discuss a propositional proof system based on a bounded degree version of Buchberger's algorithm~\cite{BUCHBERGER2006475} for finding proofs of unsatisfiability. Jefferson et al., \cite{JeffersonJGD13} use a modified form of Buchberger's algorithm that can be used to achieve the same benefits as the local-consistency algorithms which are widely used in constraint processing.

\paragraph{Applications in other proof systems.}
The bit complexity of various (semi)algebraic proof systems is another link that connects approximation algorithms and the IMP. As is easily seen, if the degree of polynomials $\vc hk$ in a proof $f_0=h_1f_1+\dots+h_kf_k$ is bounded, their coefficients can be found by representing this identity through a system of linear equations. A similar approach is used in other (semi)algebraic proof systems such as Sum-of-Squares (\Sos), in which bounded degree proofs can be expressed through an SDP program. Thus, if in addition to low degree the system of linear equations or the SDP program has a solution that can be represented with a polynomial number of bits (thus having low \emph{bit complexity}), a proof can be efficiently found. 

However, O'Donnell~\cite{ODonnell17} proved that low degree of proofs does not necessarily imply its low bit complexity. He presented a collection of polynomials that admit bounded degree \Sos\ proofs of nonnegativity, all such proofs involve polynomials with coefficients of exponential length. This means that the standard methods of solving SDPs such as the Ellipsoid Method would take exponential time to complete. Raghavendra and Weitz \cite{RaghavendraW17} suggested some sufficient conditions on combinatorial ideals that guarantee a low bit complexity \Sos\ proof exists whenever a low degree one does. Two of these conditions hold for the majority of combinatorial problems, and the third one is so-called $k$-effectiveness of the IMP part of the proof. In \cite{Bulatov20:ideal} we showed that for problems where the IMP part is of the form $\IMP(\Gm)$ (to be introduced shortly) only one of the first two conditions remains somewhat nontrivial and $k$-effectiveness can be replaced with the requirement that a variant of $\IMP(\Gm)$ is solvable in polynomial time.

\paragraph{The IMP and the CSP.}
In this paper we consider IMPs that arise from a specific class of combinatorial problems, the Constraint Satisfaction Problem or the \CSP\ for short. In a \CSP\ we are given a set of variables, and a collection of constraints on the values that variables are allowed to be assigned simultaneously. The question in a CSP is whether there is an assignment of variables that satisfies all the constraints. The CSP provides a general framework for a wide variety of combinatorial problems, and it is therefore very natural to study the IMPs that arise from constraint satisfaction problems. 

Following \cite{JeffersonJGD13,vandongenPhd,Mastrolilli19,Mastrolilli21:complexity}, every CSP can be associated with a polynomial ideal. Let CSP instance $\cP$ be given on variables $\vc xn$ that can take values from a set $D=\{0,\dots,d-1\}$. The ideal $\I(\cP)$ of $\Field[\vc xn]$ ($\Field$ is supposed to contain $D$, and it therefore usually is considered to be a numerical field, it is $\zR$ or $\zC$ in this paper) whose corresponding variety equals the set of solutions of $\cP$ is constructed as follows. First, for every $x_i$ the ideal $\I(\cP)$ contains a \emph{domain} polynomial $f_D(x_i)$ whose zeroes are precisely the elements of $D$. Then for every constraint $R(x_{i_1},\dots, x_{i_k})$, where $R$ is a predicate on $D$, the ideal $\I(\cP)$ contains a polynomial $f_R(x_{i_1},\dots, x_{i_k})$ that interpolates $R$, that is, for $(x_{i_1},\dots, x_{i_k})\in D^k$ it holds that $f_R(x_{i_1},\dots, x_{i_k})=0$ if and only if $R(x_{i_1},\dots, x_{i_k})$ is true. It is important for what follows that $\I(\cP)$ is always radical, see \cite[Lemma 8.13]{becker93grobner}. This model generalizes a number of constructions used in the literature to apply Nullstellensatz or \Sos\  proof systems to combinatorial problems, see, e.g., \cite{BeameIKPP94,BussP96,Grigoriev98,RaghavendraW17}. 

One of the major research directions in the CSP research is the study of CSPs in which the allowed types of constraints are restricted. Such restrictions are usually represented by a \emph{constraint language} that is a set of relations or predicates on a fixed set. The CSP parametrized by a constraint language $\Gm$ is denoted $\CSP(\Gm)$. 

Mastrolilli in \cite{Mastrolilli19,Mastrolilli21:complexity} initiated a systematic study of IMPs that arise from problems of the form $\CSP(\Gm)$, denoted $\IMP(\Gm)$. More precisely, for a constraint language $\Gm$ over domain $D=\{0\zd d-1\}\sse\zF$, in an instance of $\IMP(\Gm)$ we are given an instance $\cP$ of $\CSP(\Gm)$ with set of variables $X=\{\vc xn\}$, and a polynomial $f_0\in\zF[\vc xn]$. The question is whether or not $f_0$ belongs to $\I(\cP)$. Observe, that using Hilbert's Strong Nullstellensatz the problem can also be reformulated as, whether there exists a solution to $\cP$ that is not a zero of $f_0$. Sometimes we need to restrict the degree of the input polynomial, the IMP in which the degree of $f_0$ is bounded by $d$ is denoted by $\IMP_d(\Gm)$.

\paragraph{The complexity of the IMP.}
The main research question considered in \cite{Mastrolilli19,Mastrolilli21:complexity} is to classify the problems $\IMP(\Gm)$ according to their complexity. As according to the observation before this paragraph no-instances $(f_0,\cP)$ of the $\IMP(\Gm)$ can be certified by exhibiting a solution of $\cP$ that is not a zero of $f_0$, $\IMP(\Gm)$ is in co-NP. Since the only solution algorithm for $\IMP(\Gm)$ available at that point was generating a \GB, \cite{Mastrolilli19,Mastrolilli21:complexity} stated the main research problem in a different way: For which constraint languages $\Gm$ a \GB\ of $\I(\cP)$ can be constructed in polynomial time for every instance $\cP$ of $\CSP(\Gm)$? We \cite{Bulatov20:ideal} showed that the two questions are actually equivalent in all known cases. We will return to this issue later. 

Mastrolilli \cite{Mastrolilli19,Mastrolilli21:complexity} along with Mastrolilli and Bharathi \cite{Bharathi-21-DD+LIN} succeeded in characterizing the complexity of $\IMP_d(\Gm)$ for constraint languages $\Gm$ over a 2-element domain. Their results are best presented using the language of polymorphisms. A \emph{polymorphism} of a constraint language $\Gm$ over a set $D$ is a multivariate operation on $D$ that can be viewed as a multi-dimensional symmetry of relations from $\Gm$. By $\Pol(\Gm)$ we denote the set of all polymorphisms of $\Gm$. As in the case of the CSP, polymorphisms of $\Gm$ is what determines the complexity of $\IMP(\Gm)$, see \cite{Bulatov20:ideal}. For the purpose of this paper we mention three types of polymorphisms. Two of them are given by equations they satisfy. A \emph{semilattice} operation is a binary operation $f$ that satisfies the equations of idempotency $f(x,x)=x$, commutativity $f(x,y)=f(y,x)$, and associativity $f(x,f(y,z))=f(f(x,y),z)$. A \emph{majority} operation is a ternary operation $g$ that satisfies the equations $g(x,x,y)=g(x,y,x)=g(y,x,x)=x$. An important example of a majority operation is the \emph{dual-discriminator} given by $g(x,y,z)=x$ unless $y=z$, in which case $g(x,y,z)=y$. Finally, for an Abelian group $\zA$ the operation $h(x,y,z)=x-y+z$ of $\zA$ is called the \emph{affine} operation of $\zA$. 

\begin{theorem}[\cite{Mastrolilli19,Mastrolilli21:complexity,Bharathi-Minority}]\label{the:mastrolilli-intro}
Let $\Gm$ be a constraint language over $D=\{0,1\}$ and such that the \emph{constant relations} $R_0,R_1\in\Gm$, where $R_i=\{(i)\}$. Then 
\begin{itemize}
    \item[(1)]
    If $\Gm$ is invariant under a semilattice, or a majority, or affine operation (of $\zZ_2$) then $\IMP_d(\Gm)$ is polynomial time for any $d$.
    \item[(3)]
    Otherwise $\IMP_0(\Gm)$ is \coNPc.
\end{itemize}
\end{theorem}

Theorem~\ref{the:mastrolilli-intro} has been improved in \cite{Bulatov20:ideal} by showing that $\IMP_d(\Gm)$ remains polynomial time when $\Gm$ has an arbitrary semilattice polymorphism, not only on a 2-element set, an arbitrary dual-discriminator polymorphism, or an affine polymorphism of $\zZ_p$, $p$ prime.

\paragraph{Solving the IMP.}
The IMP is mostly solved using one of the two methods. The first one is the method of finding an IMP or \Sos\ proof of bounded degree using systems of linear equations or SDP programs. The other approach uses the standard polynomial division to verify whether a given polynomial has zero remainder when divided by generators of an ideal: if this is the case, the polynomial belongs to the ideal. Unfortunately the ring $\zF[\vc xn]$ is not Eucledian, which means that the result of division depends on the order in which division is performed. To avoid this complication, one may construct a \emph{Gr\"obner Basis} of the ideal that guarantees that the division algorithm always produces a definitive answer to the IMP. However, constructing a \GB\ is not always feasible, as even though the original generating set is small, the corresponding \GB\ may be huge. Even if there is an efficient way to construct a \GB, the algorithm may be quite involved and tends to work only in a limited number of cases. Moreover, in the cases where the input polynomial does not have a bounded degree, it is impossible to ensure that the division algorithm terminates in polynomial time. Even having a \GB\ with respect to graded lexicographic order (\textsf{grlex}) does not help in such case. The following example highlights this. We therefore focus on the cases where the input polynomial has a bounded degree.

\begin{example}
Let $\mc{I}$ be the ideal of $\mathbb{R}[x_1,...,x_n,y_1,...,y_n]$ generated by polynomials $x_1-y_1-1,..., x_n-y_n-1$. One can verify, this set of polynomials is a \GB\ with respect to \grlex order with $x_1 \succ x_2\succ \dots \succ x_n\succ y_1\succ \dots\succ y_n$. Now let the input polynomial be $\prod_{i=1}^n x_i $ (just one monomial, the product of all x's). If we apply the division algorithm, we obtain the expansion of the polynomial $\prod_{i=1}^n(y_i+1)$, which contains exponentially many monomials. Thus, while the division algorithm solves the problem correctly, it produces exponentially long intermediate results, and therefore is exponential time.
\end{example}

A more sophisticated approach was suggested in \cite{Bulatov20:ideal}. It involves reductions between problems of the form $\IMP(\Gm)$ before arriving to one for which a \GB\ can be constructed in a relatively simple way. Moreover, \cite{Bulatov20:ideal} also introduces a slightly different form of the IMP, called the $\chi$IMP, in which the input polynomial has indeterminates as some of its coefficients, and the problem is to find values for those indeterminates (if they exist) such that the resulting polynomial belongs to the given ideal. The results of \cite{Bulatov20:ideal} show that $\chi\IMP_d$ is solvable in polynomial time for every known case of polynomial time solvable $\IMP_d$, and that $\chi\IMP_d$ helps to find a proof of membership for $\IMP_d$.

\begin{theorem}[\cite{Bulatov20:ideal}]\label{the:ximp-intro}
\begin{itemize}
    \item[(1)]
    If $\Gm$ has a semilattice, the dual-discriminator, or the affine polymorphism of $\zZ_p$, $p$ prime, then $\chi\IMP_d(\Gm)$ is solvable in polynomial time for every $d$. 
    \item[(2)]
    If $\chi\IMP_d(\Gm)$ is polynomial time solvable then for every instance $\cP$ of $\CSP(\Gm)$ a degree $d$ \GB\ with respect to \grlex for $\I(\cP)$ can be found in polynomial time.
\end{itemize}
\end{theorem}

\subsection*{Our contribution}

\paragraph{Affine operations.}
In this paper we consider IMPs over languages invariant under affine operations of arbitrary finite Abelian groups. This type of constraint languages played a tremendously important role in the study of the CSP for three reasons. First, it captures a very natural class of problems. Problems $\CSP(\Gm)$ where $\Gm$ is invariant under an affine operation of a finite field $\zF$ can express systems of linear equations over $\zF$ and therefore often admit a classic solution algorithm such as Gaussian elimination or coset generation. In the case of a general Abelian group $\zA$ the connection with systems of linear equations is more complicated, although it is still true that every instance of $\CSP(\Gm)$ in this case can be thought of as a system of linear equations with coefficients from some ring --- the ring of endomorphisms of $\zA$ (See Proposition \ref{pro:CSP-to-LIN}).

Second, it has been observed that there are two main algorithmic approaches to solving the CSP. The first one is based on the local consistency of the problem. CSPs that can be solved solely by establishing some kind of local consistency are said to have \emph{bounded width} \cite{Bulatov08:dualities,Barto14:local}. The property to have bounded width is related to a rather surprising number of other seemingly unrelated properties. For example, Atserias and Ochremiak \cite{Atserias19:proof} demonstrated that for many standard proof systems the proof complexity of $\CSP(\Gm)$ is polynomial time if and only if $\CSP(\Gm)$ has bounded width. Also, Thapper and Zivny \cite{Thapper18:limits} discovered a connection between bounded width and the performance of various types of SDP relaxations for the Valued CSP. CSP algorithms of the second type are based on the \emph{few subpowers} property and achieve results similar to those of Gaussian elimination: they construct a concise representation of the set of all solutions to a CSP \cite{Bulatov06:simple,Idziak10:tractability}.  Problems $\CSP(\Gm)$ where $\Gm$ has an affine polymorphism were pivotal in the development of few subpowers algorithms, and, in a sense, constitute the main nontrivial case of them. Among the existing results on the IMP,  $\IMP(\Gm)$ for $\Gm$ invariant under a semilattice or majority polymorphism belong to the local consistency part of the algorithmic spectrum, while those for $\Gm$ invariant with respect to an affine operation are on the `few subpowers' part of it. It is therefore important to observe differences in approaches to the IMP in these two cases.

Third, the few subpowers algorithms \cite{Bulatov06:simple,Idziak10:tractability} when applied to systems of linear equations serve as an alternative to Gaussian elimination that also work in a more general situation and are less sensitive to the algebraic structure behind the problem. There is, therefore, a hope that studying IMPs with  an affine polymorphism may teach us about proof systems that use the IMP and do not quite work in the affine case.

The main result of this paper is 

\begin{theorem}\label{the:main-intro}
Let $\zA$ be a finite Abelian group and $\Gm$ a constraint language such that the affine operation $x-y+z$ of $\zA$ is a polymorphism of $\Gm$. Then $\IMP_d(\Gm)$ can be solved in polynomial time for any $d$. Moreover, given an instance $\cP$ of $\CSP(\Gm)$ a degree $d$ \GB\ of $\I(\cP)$ (over $\zC$) can be constructed in polynomial time.
\end{theorem}

\paragraph{The tractability of affine IMPs.}
In \cite{Bharathi-21-DD+LIN,Bulatov20:ideal,Mastrolilli19,Mastrolilli21:complexity} IMPs invariant under an affine polymorphism are represented as systems of linear equations that are first transformed to a reduced row-echelon form using Gaussian elimination, and then further converted into a \GB\ of the corresponding ideal. If $\Gm$ is a constraint language invariant under the affine operation of a general Abelian group $\zA$, none of these three steps work: an instance generally cannot be represented as a system of linear equations, Gaussian elimination does not work on systems of linear equations over an arbitrary Abelian group, and a reduced row-echelon form cannot be converted into a \GB. We therefore need to use a completely different approach. We solve $\IMP(\Gm)$ in four steps. 

First, given an instance $(f_0,\cP)$ of $\IMP_d(\Gm)$ we use the Fundamental Theorem of Abelian groups and a generalized version of pp-interpretations for the IMP \cite{Bulatov20:ideal} to reduce $(f_0,\cP)$ to an instance $(f'_0,\cP')$ of \emph{multi-sorted}  $\IMP(\Dl)$ (see below), in which every variable takes values from a set of the form $\zZ_{p^\ell}$, $p$ prime, where the parameters $p,\ell$ may be different for different variables. Second, we show that $\cP'$ can be transformed in such a way that constraints only apply to variables having the same domain. This means that $\cP'$ can be thought of as a collection of disjoint instances $\vc{\cP'}k$ where $\cP'_i$ has the domain $\zZ_{p_i^{\ell_i}}$, $p_i$ prime. If we only had to solve the CSP, we could just solve these instances separately and report that a solution exists if one exists for each $\cP_i'$. This is however not enough to solve the IMP, as the polynomial $f'_0$ involves variables from each of the instances $\cP'_i$. Therefore, the third step is to generate some form of a system of linear equations and its reduced row-echelon form for each $\cP'_i$ separately. For each subinstance $\cP'_i$ we identify a set of variables, say, $\vc xr$, that can take arbitrary values from $\zZ_{p_i^{\ell_i}}$. For each of the remaining variables $x$ of $\cP'_i$ we construct an expression 
\begin{equation}\label{equ:p-expression}
x=\al_1x_1 \oplus \dots \oplus \al_r x_r \pmod {p_i^{\ell_i}}
\end{equation} such that for any choice of values of $\vc xr$ assigning the remaining variables according to these expressions results in a solution of $\cP'_i$. Unfortunately, expressions like this cannot be easily transformed into a \GB\ --- the entire paper \cite{Bharathi-Minority} is devoted to such transformation in the case of $\zZ_2$. The reason for that is that modular linear expressions can only be interpolated by polynomials of very high degree, and therefore if such an expression contains many variables, an interpolating polynomial is exponentially long. In order to convert constraints of $\cP'_i$ into a \GB\ we need one more step. This step is changing the domain of the problem. Instead of the domain $\zZ_{p_i^{\ell_i}}$ we map the problem to the domain $U_{p_i^{\ell_i}}$ of $p_i^{\ell_i}$-th roots of unity. This allows us to replace modular addition with complex multiplication, and so \eqref{equ:p-expression} is converted into a polynomial with just two monomials. The input polynomial $f'_0$ is transformed accordingly. Finally, we prove that the resulting polynomials form a \GB, and so the problem can now be solved in the usual way.

\paragraph{Multi-sorted IMPs.}
In order to prove Theorem~\ref{the:main-intro} we introduce two techniques new to the IMP research, although the first one has been extensively used for the CSP. The first technique is multi-sorted problems mentioned above, in which every variable has its own domain of values. This framework is standard for the CSP, and also works very well for the IMP, as long as the domain of each variable can be embedded into the field of real or complex numbers. However, many concepts used in proofs and solution algorithms such as pp-definitions, pp-interpretations, polymorphisms have to be significantly adjusted, and several existing results have to be reproved in this more general setting. However, in spite of this extra work, the multi-sorted IMP may become the standard framework in this line of research.

\paragraph{A general approach to $\chi$IMP.}
In \cite{Bulatov20:ideal} we introduced \xIMP, a variation of the IMP, in which given a CSP instance $\cP$ and a polynomial $f_0$ some of whose coefficients are unknown, the goal is to find values of the unknown coefficients such that the resulting polynomial $f'_0$ belongs to $\I(\cP)$; or report such values do not exist. This framework has been instrumental in finding a \GB\ and therefore finding a proof of membership when $f_0$ has bounded degree, as well as in establishing connections between the IMP and other proof systems such as \Sos. We again use $\chi$IMP to prove the second part of Theorem~\ref{the:main-intro}. The key idea is that whenever the bounded-degree search version of $\chi$IMP is solvable in polynomial time, one can construct a bounded-degree \GB. We therefore focus on instances of $\chi$IMP that arise from CSPs over finite Abelian groups and are invariant under the affine polymorphism. To this end, we refine the approach of constructing bounded-degree \GB\ via $\chi$IMP in two ways (see Section~\ref{sec:sub}). First, we adapt it for multi-sorted problems. Second, while in \cite{Bulatov20:ideal} reductions for $\chi$IMP are proved in an ad hoc manner, here we develop a unifying construction based on substitution reductions that covers all the useful cases so far.

\paragraph{Organization of the paper.} The paper is organized as follows. In \Cref{sec:prelim} we provide the preliminaries and the notation we use throughout the paper. In \Cref{sec:multi-sorted-reductions} we focus on multi-sorted constraint languages and their corresponding CSPs. We show that in the multi-sorted case, similar to the one-sorted case, there are reductions for \IMP s with respect to pp-definability and pp-interpretability. These reductions are instrumental for us, as in \Cref{sec:Abelian-csps} we show that any constraint language invariant under an affine operation of some finite Abelian group can be pp-interpreted by a multi-sorted constraint language over very simple groups, where the CSPs can be transformed into systems of linear equations over these simple groups. Having these transformations, we begin by proving the decision version $\IMP_d$ in \Cref{sec:solving-imp}—arguably the easier task. This is done using a substitution technique that formulates the problem over roots of unity. This result does not necessarily give us a way of constructing degree $d$ \GB\ for the original problem and can only help us with the decision version. We, however, in \Cref{sec:sub}, use our substitution technique to present a method for computing a degree $d$ \GB. This, in turn, allows us not only to solve $\IMP_d$ but also to construct a membership proof, when one exists.

\section{Preliminaries}
\label{sec:prelim}
\subsection{Ideals and varieties}

Let $\Field$ denote an arbitrary field. Let $\Field[x_1,\dots, x_n]$ be the ring of polynomials over the field $\Field$ and indeterminates $x_1, \dots,x_n$. Sometimes it will be convenient not to assume any specific ordering or names of the indeterminates. In such cases we use $\Field[X]$ instead, where $X$ is a set of indeterminates, and treat points in $\Field^X$ as mappings $\vf:X\to\Field$. The value of a polynomial $f\in\Field[X]$ is then written as $f(\vf)$. Let $\Field[x_1,\dots, x_n]_d$ denote the subset of polynomials of degree at most $d$. An \emph{ideal} of $\Field[x_1,\dots, x_n]$ is a set of polynomials from $\Field[x_1,\dots, x_n]$ closed under addition and multiplication by a polynomial from $\Field[x_1,\dots, x_n]$. 

We will need ideals represented by a generating set. The ideal (of $\Field[x_1,\dots, x_n]$) generated by a finite set of polynomials $\{f_1, \dots,f_m\}$ in $\Field[x_1,\dots, x_n]$ is defined as
    \[
        \mb{I}(f_1, \dots,f_m)\overset{\mathrm{def}}{=}\Big\{ \sum\limits_{i=1}^m t_if_i \mid t_i \in \Field[x_1,\dots, x_n]\Big\}.
    \]
Another common way to denote $\mb{I}(f_1,\ldots,f_m)$ is by $\langle f_1,\ldots,f_m \rangle$ and we will use both notations interchangeably. For a set of points $S\sse\Field^n$ its \emph{vanishing ideal} is the set of polynomials defined as
    \[
        \mb{I}(S) \overset{\mathrm{def}}{=} \{f \in \Field[x_1,\dots, x_n] : f(a_1,\dots,a_n) = 0 \ \ \forall (a_1,\dots,a_n) \in S\}.
    \]
The \emph{affine variety} defined by a set of polynomials $\{f_1,\ldots, f_m\}$ is the set of common zeros of $\{f_1,\ldots, f_m\}$, 
    \[
        \Variety{ f_1,\ldots,f_m}\overset{\mathrm{def}}{=} \{(a_1,\ldots,a_n)\in \Field^n \mid  f_i(a_1,\ldots,a_n)=0 \quad 1\leq i\leq m\}.
    \]
Similarly, for an ideal $\I\subseteq \Field[x_1,\ldots,x_n]$ its affine variety is the set of common zeros of all the polynomials in $\I$. This is denoted by $\Variety{\I}$ and is formally defined as 
    \[
        \Variety{\I}=\{(a_1,\ldots,a_n)\in \Field^n\mid f(a_1,\ldots,a_n)=0 \quad \forall f\in \I\}.
    \]
In the case where $\Field$ is algebraically closed, for instance it is the field $\Complex$ of complex numbers, one can guarantee that the only ideal which represents the empty variety is the entire polynomial ring itself i.e., if $\Variety{\I} = \emptyset$ then $\I = \Field[\vc x n]$. This is known as the \emph{Weak Nullstellensatz}. 




As we will discuss it, the ideals arising from \CSP s are radical ideals. An ideal $\I$ is \emph{radical} if $f^m \in\I$ for some integer $m\geq 1$ implies that $f\in \I$. For an arbitrary ideal $\I$ the smallest radical ideal containing $\I$ is denoted $\sqrt\I$. In other words $\sqrt\I=\{f\in \Field[x_1,\dots, x_n]\mid f^m\in\I\text{ for some $m$}\}$. Radical ideals have a strong connection to affine varieties, as they consist of all polynomials which vanish on some variety $V$. This one-to-one correspondence between affine varieties and radical ideals is known as the \emph{Strong Nullstellensatz}. That is, for an algebraically closed field $\Field$ and an ideal $\I\subseteq \Field[\vc x n]$ we have $\mathbf{I}(\Variety{\I})=\sqrt{\I}$.






Finally, the following theorem is a useful tool for finding ideals corresponding to intersections of varieties. We will use it in the following subsections where we construct ideals corresponding to \CSP~instances, and it will also be used in our proofs. Here, the sum of ideals $\I$ and $\J$, denoted $\I+\J$, is the set $\I+\J=\{f+g\mid f\in\I \text{ and } g\in\J\}$.

\begin{theorem}[\cite{Cox}, Theorem 4, p.190]\label{th:ideal_intersection}
  If $\I$ and $\J$ are ideals in $\Field[x_1,\ldots,$ $x_n]$, then $\Variety{\I+\J}=\Variety{\I}\cap\Variety{\J}$ and $\Variety{\I\cap \J}= \Variety{\I}\cup \Variety{\J}$.
\end{theorem}

\subsection{The Constraint Satisfaction Problem}

We use $[n]$ to denote $\{1\zd n\}$. Let $D$ be a finite set, it will often be referred to as a domain. An \emph{$n$-ary relation} on $D$ is a set of $n$-tuples of elements from $D$; we use $\mb{R}_D$ to denote the set of all finitary relations on $D$. 
A \emph{constraint language} is a subset of $\mb{R}_D$, and may be finite or infinite. Given any constraint language, we can define the corresponding class of
 constraint satisfaction problems in the following way.
 
 \begin{definition}
    Let $\Gm\subseteq \mb{R}_D$ be a constraint language over domain $D$. The  constraint satisfaction problem over $\Gm$, denoted $\CSP(\Gm)$, is defined to be the decision problem with
    \begin{description}
        \item[Instance:] An instance $\cP = (X, D, \mc{C})$ where $X$ is a finite set of variables, and $\mc{C}$ is a set of constraints where each constraint $C\in \mc{C}$ is pair $\langle \bs, R\rangle$, such that
        \begin{itemize}
            \item[-] $\bs=(x_{i_1},\ldots,x_{i_k})$ is a list of variables of length $k$ (not necessarily distinct), called the \emph{constraint scope};
            \item[-] $R$ is a $k$-ary relation on $D$, belonging to $\Gamma$, called the \emph{constraint relation}.
        \end{itemize}
        \item[Question:] Does there exist a \emph{solution}, i.e., a mapping $\vf: X \to D$ such that for each constraint $\langle \bs, R\rangle\in\mc{C}$, with $\bs=(x_{i_1},\ldots,x_{i_k})$, the tuple $(\vf(x_{i_1}),\dots,\vf(x_{i_k}))$ belongs to $R$ ?
    \end{description}
 \end{definition}
We point out that another common way to denote a constraint $\langle \bs, R\rangle$ is by $R(\bs)$, that is, to treat $R$ as a predicate, and we will use both notations interchangeably. Moreover,  we will use $\Sol(\cP)$ to denote the (possibly empty) set of solutions of $\cP$.




\subsection{The ideal-CSP correspondence}\label{sect:idealCSP}

Here, we explain how to construct an ideal corresponding to a given instance of \CSP. Constraints are in essence varieties, see e.g.~\cite{JeffersonJGD13,vandongenPhd}. Following~\cite{Mastrolilli19,vandongenPhd}, we shall translate CSPs to polynomial ideals and back. Let $\cP=(X,D,C)$ be an instance of $\CSP(\Gamma)$ where $\Gm$ is a fixed constraint language with relations of fixed arities. Without loss of generality, we assume that $D\subset\Field$. In fact, we will mainly assume $\Field=\Real$ or $\Field=\Complex$, and $D=\{0,1,\dots,|D|-1\}$. Let $\Sol(\cP)$ be the (possibly empty) set of all solutions of $\cP$. We wish to map $\Sol(\cP)$ to an ideal $\I(\cP)\subseteq \Field[X]$ such that $\Sol(\cP)=\Variety{\I(\cP)}$.

First, for every $x_i$ the ideal $\I(\cP)$ contains a \emph{domain} polynomial $f_D(x_i)$ whose zeroes are precisely the elements of $D$. Then for every constraint $R(x_{i_1},\dots, x_{i_k})$, where $R$ is a predicate on $D$, the ideal $\I(\cP)$ contains a polynomial $f_R(x_{i_1},\dots, x_{i_k})$ that interpolates $R$, that is, for $(x_{i_1},\dots, x_{i_k})\in D^k$ it holds that $f_R(x_{i_1},\dots, x_{i_k})=0$ if and only if $R(x_{i_1},\dots, x_{i_k})$ is true. Note that each $f_R$ has bounded degree, this is because $D$ and $k$ are fixed. 

Including a domain polynomial for each variable has the advantage that it ensures that the ideals generated by our systems of polynomials are radical (see Lemma 8.19 of \cite{becker93grobner}). Hence, by Weak and Strong Nullstellensatz, we have the following properties.

\begin{theorem}\label{th:nullstz}
Let $\cP$ be an instance of the $\CSP(\Gamma)$ and $\I(\cP)$ constructed as above. Then
  \begin{align}
    &\Variety{\I(\cP)}=\emptyset \Leftrightarrow 1\in \I(\cP) \Leftrightarrow \I(\cP)=\Field[X],  \tag{Weak Nullstellensatz}\label{eq:weak_nstz}\\
    &\I(\Variety{\I(\cP)})=\sqrt{\I(\cP)},\tag{Strong Nullstellensatz}\label{eq:strong_nstz}\\
    &\sqrt{\I(\cP)}=\I(\cP).\tag{Radical Ideal}\label{eq:ICradical}
  \end{align}
\end{theorem}

\subsection{The Ideal Membership Problem}

In the general Ideal Membership Problem we are given an ideal $\I\sse\Field[\vc xn]$, usually by some finite generating set, and a polynomial $f_0$. The question then is to decide whether or not $f_0\in\I$. If $\I$ is given through a CSP instance, we can be more specific.

\begin{definition}\label{def:imp}
The {\sc Ideal Membership Problem} associated with a constraint language $\Gamma$  over a set $D$ is the problem $\IMP(\Gamma)$ in
which the input is a pair $(f_0,\cP)$ where $\cP = (X, D, C)$ is a
$\CSP(\Gm)$ instance and $f_0$ is a polynomial from $\Field[X]$. The goal is to decide
whether $f_0$ lies in the ideal $\I(\cP)$. We use
$\IMP_d(\Gamma)$ to denote $\IMP(\Gamma)$ when the input polynomial
$f_0$ has degree at most $d$.
\end{definition}

As $\I(\cP)$ is radical, by the Strong Nullstellensatz an equivalent way to solve the membership problem $f_0 \in \I(\cP)$ is to answer the following question:

\begin{quote}
Does there exist an $\mb{a}\in \Variety{\I(\cP)}$ such that $f_0(\mb{a})\neq 0$?
\end{quote}
In the \textbf{yes} case, such an $\mb{a}$ exists if and only if $f_0\not\in \mb{I}(\Variety{\I(\cP)})$ and therefore $f_0$ is \textbf{not} in the ideal $\I(\cP)$. This observation also implies that for any constraint language $\Gm$ the problem $\IMP_0(\Gm)$ is equivalent to $\mathsf{not\text{-}}\CSP(\Gm)$ \cite{Bulatov20:ideal}. Moreover, note that $\IMP(\Gm)$ belongs to \textbf{coNP} for any $\Gm$ over a finite domain. We say that $\IMP(\Gm)$ is \emph{tractable} if it can be solved in polynomial time. We say that $\IMP(\Gm)$ is \emph{$d$-tractable} if $\IMP_d(\Gm)$ can be solved in polynomial time for every $d$.\footnote{Almost all algorithmic results on $\IMP(\Gm)$ prove $d$-tractability for arbitrary $d$.  We are not aware of a constraint language fo which $\IMP(\Gm)$ is $d$-tractable but provably not tractable. However, Mastrolilli \cite[Remark 5.1]{Mastrolilli21:complexity} provided an example of a $\Gm$, for which $\IMP(\Gm)$ is $d$-tractable, which suggests that the degree of the reduced GB can be arbitrarily large.} 

Throughout this paper by the search $\IMP$ we understand the following problem. 
\begin{quote}
\label{search-IMP}
{\bf Search Version of \IMP.} 
Let $(f_0,\cP)$ be an instance of $\IMP(\Gm)$ such that $f_0\not\in\I(\cP)$, the problem is to find an assignment $\vf\in \mb{V}(\I(\cP))$ such that $f_0(\vf)\neq 0$.
\end{quote}

Moreover, throughout this paper whenever we say finding a proof of membership we mean the following problem. 
\begin{quote}
\label{search-IMP}
{\bf Finding a Membership Proof.} 
Let $(f_0,\cP)$ be an instance of $\IMP(\Gm)$ such that $f_0\in\I(\cP)$, the problem is to find polynomials $h_1,\dots,h_k\in\Field[x_1,\dots,x_n]$ such that $f_0=\sum_{i=1}^k h_ip_i$ where $\I(\cP)=\langle p_1,\dots,p_k\rangle$.
\end{quote}


\subsection{The Ideal Membership Problem and \GBs}
 We use the standard notation from algebraic geometry and follow notation in \cite{Cox}. A monomial ordering $\succ$ on $\Field[\vc x n]$ is a relation $\succ$ on $\zZ_{\geq 0}^n$, or equivalently, a relation on the set of monomials $\bx^{\alpha}$, $\alpha \in \zZ_{\geq 0}^n$ (see~\cite{Cox}, Definition 1, p.55).  Each monomial $\bx^\alpha=x_1^{\alpha_1}\cdots x_n^{\alpha_n}$ corresponds to an $n$-tuple of exponents $\alpha =(\alpha_1,\ldots,\alpha_n)\in \mathbb{Z}^n_{\geq0}$. In this paper we use two standard monomial orderings, namely \emph{lexicographic order (\lex)} and \emph{graded lexicographic order (\grlex)}. Let $\alpha =(\alpha_1,\ldots,\alpha_n),\beta=(\beta_1,\ldots,\beta_n)\in \mathbb{Z}^n_{\geq0}$ and $|\alpha| = \sum_{i=1}^n\alpha_i$, $|\beta| = ~\sum_{i=1}^n\beta_i$. We say $\alpha\succ_\lex \beta$ if the leftmost nonzero entry of the vector difference $\alpha -\beta \in \mathbb{Z}^n$ is positive. We will write $\bx^\alpha\succ_\lex \bx^\beta$ if $\alpha\succ_\lex\beta$. We say $\alpha\succ_\grlex \beta$ if $|\alpha| >|\beta|$, or $|\alpha| =|\beta|$ and $\alpha\succ_\lex \beta$.

\begin{definition}
Let $f= \sum_{\alpha} a_{\alpha}\bx^\alpha$ be a nonzero polynomial in $\Field[x_1,\ldots,x_n]$ and let $\succ$ be a monomial order.
  \begin{enumerate}
    \item The \emph{multidegree} of $f$ is $\multideg(f)\overset{\mathrm{def}}{=} \max(\alpha\in \mathbb{Z}^n_{\geq0}:a_\alpha\not = 0)$.
    \item The \emph{degree} of $f$ is deg$(f)=|\multideg(f)|$ where $|\alpha| = \sum_{i=1}^n\alpha_i$. In this paper, degree is always according to \textsf{grlex} order. 
    \item The \emph{leading coefficient} of $f$ is $\LC(f)\overset{\mathrm{def}}{=} a_{\multideg(f)}\in \Field$.
    \item The \emph{leading monomial} of $f$ is $\LM(f)\overset{\mathrm{def}}{=} \bx^{\multideg(f)}$ (with coefficient 1).
    \item The \emph{leading term} of $f$ is $\LT(f)\overset{\mathrm{def}}{=} \LC(f)\cdot \LM(f)$.
  \end{enumerate}
\end{definition}


The Hilbert Basis Theorem states that every ideal has a finite generating set (see, e.g., Theorem~4 on page 77 \cite{Cox}). A particular type of finite generating set is the \GBs\ which are quite well-behaved in terms of division. 


\begin{definition}\label{def:GB}
   Fix a monomial order on the polynomial ring $\Field[\vc x n ]$. A finite subset $G = \{\vc g t\}$ of an ideal $\I \subseteq \Field[\vc x n ]$ different from $\{0\}$ is said to be a \GB\ (or \emph{standard basis}) if 
   \[
    \Ideal{\LT(g_1),\dots,\LT(g_t)} = \Ideal{\LT(I)}\]
    where $\Ideal{\LT(I)}$ denotes the ideal generated by the leading terms of elements of $\I$. 
\end{definition}
In the above definition if we restrict ourselves to the polynomials of degree at most $d$ then we obtain the so called \emph{$d$-truncated \GB }. The $d$-truncated \GB\ $G_d$ of $G$ is defined as $G_d =G\cap \Field[\vc x n]_d$. The main appeal of \GBs\ is that the remainder of division by a \GB\ is uniquely defined, no matter in which order we do the division.
\begin{proposition}[\cite{Cox}, Proposition~1, p.83 ]
\label{prop:division-GB}
    Let $\I \subseteq \Field[\vc x n ]$ be an ideal and let $G = \{\vc g t\}$ be a \GB\ for $\I$. Then given $f \in \Field[\vc x n ]$, there is a unique $r \in \Field[\vc x n ]$ with the following two properties:
    \begin{enumerate}
        \item No term of $r$ is divisible by any of $\LT(g_1),\dots,\LT(g_t)$,
        \item There is $g\in\I$ such that $f=g+r$.
    \end{enumerate}
\end{proposition}

In the above definition, $r$ is called the remainder on division of $f$ by $G$ and it is also known as the \emph{normal form of} $f$ \emph{by} $G$, denoted by $f|_G$. As a corollary of \Cref{prop:division-GB}, we get the following criterion for checking the ideal membership.
\begin{corollary}[\cite{Cox}, Corollary~2, p.84]
    Let $G = \{\vc g t\}$ be a \GB\ for an ideal $\I \subseteq \Field[\vc x n]$ and let $f \in \Field[\vc x n]$. Then $f \in \I$ if and only if the remainder on division of $f$ by $G$ is zero.
\end{corollary}


Buchberger \cite{BUCHBERGER2006475} introduced a criterion known as the Buchberger's criterion to verify if a set of polynomials is a \GB. In order to formally express this criterion, we need to define the notion of $S$-polynomials. Let $f,g\in \Field[x_1,\ldots,x_n]$ be nonzero polynomials. If $\multideg(f)=\alpha$ and $\multideg(g)=\beta$, then let $\gamma=(\gamma_1,\ldots,\gamma_n)$, where $\gamma_i = \max(\alpha_i,\beta_i)$ for each $i$. We call $x^\gamma$ the \emph{least common multiple} of $\LM(f)$ and $\LM(g)$, written $x^\gamma = \LCM(\LM(f),\LM(g))$. The \emph{$S$-polynomial} of $f$ and $g$ is the combination 
    \[
        S(f,g) = \frac{x^\gamma}{\LT(f)}\cdot f - \frac{x^\gamma}{\LT(g)}\cdot g.
    \]


 
 \begin{theorem}[Buchberger's Criterion \cite{Cox}, Theorem 3, p.105]
 \label{th:crit}
     Let $\I$ be a polynomial ideal. Then a basis $G = \{\vc g t\}$ of $\I$ is a \GB\ of $\I$ if and only if for all pairs $i\neq j$, the remainder on division of $S(g_i, g_j)$ by $G$ (listed in some order) is zero.
 \end{theorem}

\begin{proposition}[\cite{Cox}, Proposition 4, p.106]
\label{prop:prime-LM}
    We say the leading monomials of two polynomials $f,g$ are relatively prime if $\LCM(\LM(f),\LM(g))=\LM(f)\cdot\LM(g)$. Given a finite set $G \subseteq \Field[\vc x n]$, suppose that we have $f,g \in G$ such that the leading monomials of $f$ and $g$ are relatively prime. Then the remainder on division of $S(f, g)$ by $G$ is zero.
\end{proposition}


\section{Multi-sorted CSPs and IMP}
\label{sec:multi-sorted-reductions}

\subsection{Multi-sorted problems}
In most theoretical studies of the \CSP\ all variables are assumed to have the same domain, \CSP s of this type are known as \emph{one-sorted \CSP s}. However, for various purposes, mainly for more involved algorithms such as in \cite{Bulatov17,Zhuk17,Zhuk20:proof} one might consider \CSP s where different variables of a \CSP\ have different domains, \CSP s of this type are known as \emph{multi-sorted \CSP s}~\cite{BulatovJ03-multi-sorted}. We study this notion in the context of the \IMP\ and provide a reduction for multi-sorted languages that are pp-interpretable. This in particular is useful in this paper as it provides a reduction between languages that are invariant under an affine polymorphism over an arbitrary Abelian group and languages over several cyclic $p$-groups. Definitions below are from \cite{BulatovJ03-multi-sorted}.

\begin{definition}
    For any finite collection of finite domains $\mc{D} =\{D_t\mid t \in T\}$, and any list of indices $(t_1, t_2,\dots, t_m) \in T^m$, a subset $R$ of $D_{t_1}\times D_{t_2} \times\dots\times D_{t_m}$, together with the list $(t_1, t_2,\dots, t_m)$, is called a \emph{multi-sorted} relation over $\mc{D}$ with arity $m$ and signature  $(t_1, t_2,\dots, t_m)$. For any such relation $R$, the signature of $R$ is
    denoted $\sigma(R)$.
\end{definition}

As an example consider $\mc{D}=\{D_1,D_2\}$ with $D_1=\{0,1\}$, $D_2=\{0,1,2\}$. Then $\zZ_6$, which is the direct sum of $\zZ_2$ and $\zZ_3$, $\zZ_2\oplus\zZ_3$, can be viewed as a multi-sorted relation over $\mc{D}$ of arity $2$ with signature $(1,2)$. 

Similar to regular one-sorted \CSP s, given any set of multi-sorted relations, we can define a corresponding class of \CSP s. Let $\Gm$ be a set of multi-sorted relations over a collection of sets $\mc{D} =\{D_t\mid t \in T\}$. The multi-sorted constraint satisfaction problem over $\Gm$, denoted as before $\MCSP(\Gm)$, is defined to be the decision problem with instance $\mc{P} = (X, \mc{D}, \delta, \mc{C})$, where $X$ is a finite set of variables, $\delta:X\to T$, and  $\mc{C}$ is a set of constraints where each constraint $C \in \mc{C}$ is a pair $\langle \mathbf{s},R\rangle$, such that 
\begin{itemize}
    \item $\mathbf{s} = (x_1, \dots, x_{m_C})$ is a tuple of variables of length $m_C$, called the constraint scope;
    \item $R$ is an element of $\Gm$ with arity $m_C$ and signature $(\delta(x_1),\dots,\delta(x_{m_c}))$,
    called the constraint relation.
\end{itemize}
The goal is to decide whether or not there exists a solution, i.e.\ a mapping $\varphi: X \to \cup_{D\in\mc{D}}D$, with $\varphi(x)\in D_{\delta(x)}$, satisfying
all of the constraints. We will use $\Sol(\mc{P})$ to denote the (possibly empty) set of solutions of the instance $\mc{P}$. 

The multi-sorted IMP, that is, $\IMP(\Gm)$ for a multi-sorted constraint language $\Gm$ is largely defined in the same way as the regular one. The ideal corresponding to an instance $\mc{P}$ of $\MCSP(\Gm)$ is constructed similar to the one-sorted case, the only difference is that for an instance $\mc{P} = (X, \mc{D}, \delta, \mc{C})$ the corresponding ideal $\I(\mc{P})$ contains domain polynomials $\prod_{a\in D_{\delta(x_i)}}(x_i-a)$ for each variable $x_i$. As with one-sorted \CSP s, the \IMP\ associated with a multi-sorted constraint language $\Gm$ over a set $\mc{D}$ is the problem
$\IMP(\Gm)$ in which the input is a pair $(f,\mc{P})$ where  $\mc{P} = (X, \mc{D}, \delta, \mc{C})$ is a $\MCSP(\Gm)$ instance and $f$ is a polynomial from $\Field[X]$. The goal is to decide whether $f$ lies in the ideal $\I(\mc{P})$. We use $\IMP_d(\Gm)$
to denote $\IMP(\Gm)$ when the input polynomial $f$ has degree at most $d$.

\subsection{Primitive-positive definability}

Primitive-positive (pp-) definitions have proved to be instrumental in the study of the CSP \cite{JeavonsCG97,BulatovJK05} and of the IMP as well \cite{Bulatov20:ideal}. Here we introduce the definition of pp-definitions and the more powerful construction, pp-interpretations, in the multi-sorted case, and prove that, similar to the one-sorted case \cite{Bulatov20:ideal}, they give rise to reductions between IMPs.

\begin{definition}[pp-definability]
\label{def:pp-def}
Let $\Gm$ be a multi-sorted constraint language on a collection of sets $\mc D=\{D_t\mid t\in T\}$. A \emph{primitive-positive (pp-) formula} in the language $\Gm$ is a first order formula $L$ over variables $X$ that uses predicates from $\Gm$, equality relations, existential quantifier, and conjunctions, and satisfies the condition:
\begin{quote}
    if $R_1(\vc xk),R_2(\vc y\ell)$ are atomic formulas in $L$ with signatures $\sg_1,\sg_2$ and such that $x_i,y_j$ are the same variable, then $\sg_1(i)=\sg_2(j)$.
\end{quote}
The condition above determines the signature $\sg:X\to T$ of $L$.

Let $\Dl$ be another multi-sorted language over $\mc D$. We say that $\Gm$ pp-defines $\Dl$ (or $\Dl$ is pp-definable from $\Gm$) if for each ($k$-ary) relation (predicate) $R\in\Dl$ there exists a pp-formula $L$ over variables $\{x_1,\dots, x_m,x_{m+1},\linebreak \dots,x_{m+k}\}$ such that 
    \[
        R(x_{m+1},\dots,x_{m+k})= \exists x_1 \dots \exists x_m  L',
    \]
where $L=\exists x_1 \dots \exists x_m  L'$, and if $\sg,\sg'$ are the signatures of $L$ and $R$, respectively, then $\sg'=\sg_{|\{m+1\zd m+k\}}$.
\end{definition}

An analog of the following result for one-sorted \CSP s is proved in \cite{Bulatov20:ideal}. 

\begin{theorem}
\label{thm:Multi-sorted-pp-reduction}
    If multi-sorted constraint language $\Gm$ pp-defines multi-sorted constraint language $\Dl$, then $\IMP(\Dl)$ [$\IMP_d(\Dl)$] is polynomial time reducible to $\IMP(\Gm)$ [respectively, to $\IMP_d(\Gm)$].
\end{theorem}

We provide a proof for Theorem~\ref{thm:Multi-sorted-pp-reduction}. Our proof is a slight modification of the one given for one-sorted \CSP s in \cite{Bulatov20:ideal}. In order to establish the result, we first analyze the relationship between pp-definability and the notion of elimination ideal from algebraic geometry. This has been explored in the case of one-sorted CSPs in \cite{Bulatov20:ideal,Mastrolilli19}. Here we establish a relationship in the case of multi-sorted CSPs. The proofs and ideas used here are almost identical to the ones in \cite{Mastrolilli19}. 

\begin{definition}
Given $\I=\langle f_1,\dots,f_s\rangle \subseteq \Field[X]$, for $Y\sse X$, the $Y$-elimination ideal $\I_{X\setminus Y}$ is the ideal of $\Field[X\setminus Y]$ defined by 
\[
\I_{X\setminus Y} =\I\cap \Field[X\setminus Y]
\]
In other words, $\I_{X\setminus Y}$ consists of all consequences of $f_1 = \dots = f_s = 0$ that do not depend on variables from $Y$.
\end{definition}

\begin{theorem}\label{extension-theorem}
    Let $\cP=(X,\mc{D},\delta,C)$ be an instance of the CSP$(\Gamma)$, and let $\I(\cP)$ be its corresponding ideal. For any $Y\sse X$ let $\I_Y$ be the $(X\setminus Y)$-elimination ideal. Then, for any partial solution $\vf_Y\in \mb{V} (\I_Y)$ there exists an extension $\psi: X\setminus Y\to \cup_{t\in T}D_t$ such that $(\vf,\psi) \in \mb{V} (\I(\cP))$.
\end{theorem}
\begin{proof}
    Suppose $\mc{D}=\{D_t\mid t\in T\}$, $\delta:X\to T$ and set $\I=\I(\mc{P})$. We may assume $\Variety{\I}\neq \emptyset$. This is because if $\Variety{\I}=\emptyset$ then $1\in \I$ which implies $1\in \I_Y$. If the latter holds then the claim is vacuously true. We proceed by assuming $\Variety{\I}\neq \emptyset$ and $\Variety{\I_Y}\neq \emptyset$.
    
    The proof is by contradiction. Suppose there exists $\ba=(\vc am)\in \Variety{\I_Y}$, with $m=|Y|$, that does not extend to a feasible solution from $\Variety{\I}$. Assume $Y=\{\vc ym\}\subseteq X$, and define the polynomial 
    \begin{align*}
        q(\vc ym) = \prod_{i=1}^m~\prod_{j\in D_{\delta(y_i)}\setminus \{a_i\}}(y_i-j).
    \end{align*}
    Observe that $q(\vc am)\neq 0$ however for every $\bb=(\vc bm)$ that can be extended to a feasible solution of $\Variety{\I}$ we have $q(\bb)=0$. This implies
    \begin{align*}
        q(\vc ym)\in \mathbf{I}(\Variety{\I})\cap \Field[\vc ym]=\I\cap \Field[\vc ym] = \I_Y
    \end{align*}
    where the first equality follows from the Strong Nullstellensat and $\I$ being radical. Having $q(\vc ym)\in \I_Y$ implies that $\ba=(\vc am)\not\in\Variety{\I_Y}$, a contradiction.
\end{proof}

Let $\Dl$ and $\Gm$ be multi-sorted constraint languages over $\mc{D}$ where $\Gm$ pp-defines $\Dl$. That is for each ($k$-ary) relation (predicate) $R\in\Dl$ there exists a quantifier-free conjunctive formula $L$ (the quantifier-free part of a pp-definition of $R$ over variables $\{x_1,\dots, x_m,x_{m+1},\dots,x_{m+k}\}$ that uses predicates from $\Gm$, equality relations, and conjunctions such that  
    \begin{align}
    \label{eq:pp-def}
        R(x_{m+1},\dots,x_{m+k})= \exists x_1 \dots \exists x_m  L,
    \end{align}
and if $\sg,\sg'$ are the signatures of $L$ and $R$, respectively, then $\sg'=\sg_{|\{m+1\zd m+k\}}$. Let $S=\Sol(L)$ be the set of satisfying assignments for $L$ and let $\mathbf{I}(S)$ be its corresponding vanishing ideal. Note that $\mathbf{I}(S)\subseteq \Field[\vc x{m+k}]$.

\begin{lemma}
\label{variety-m-elimination}
$R = \Variety{\I_X}$ where $\I_X = \mathbf{I}(S)\cap \Field[x_{m+1},\dots,x_{m+k}]$.
\end{lemma}
\begin{proof}
     Define the mapping $\pi_X:\Field^{m+k}\to \Field^k$ to be the projection 
     \[
        \pi_X(\vc a{m+k}) = (a_{m+1},\dots,a_{m+k}).    
     \]
     If we apply $\pi_X$ to $S$ we get $\pi_X(S)\subseteq \Field^k$. Now it is easy to see $\pi(S)\subseteq \Variety{\I_X}$. This is because every polynomial in $\I_X$ vanishes on all the points in $\pi(S)$. Provided that $\pi(S)\subseteq \Variety{\I_X}$ we can write $R$ as follows
     \[
        R=\pi_X(S) =\{(a_{m+1},\dots,a_{m+k})\in \Variety{\I_X}\mid \exists \vc am \in \Field \text{ such that } (\vc a{m+k})\in S\}
     \]
     This is exactly the set of points in $\Variety{\I_X}$ that can be extended to a solution in $S$. However, by Theorem~\ref{extension-theorem} all the points in $\Variety{\I_X}$ can be extended to a solution in $S$. This means $R=\Variety{\I_X}$ as desired.
\end{proof}
We now have all the required ingredients to prove Theorem~\ref{thm:Multi-sorted-pp-reduction}.

\begin{proof}[Proof of Theorem~\ref{thm:Multi-sorted-pp-reduction}]
     At a high level, in this proof an instance of $\CSP(\Delta)$ is transformed into an instance of $\CSP(\Gamma)$ in the standard
    way and the input polynomial is kept the same. Next we provide the details. Let $(f,\cP_\Dl)$, $\cP_\Dl=(X,\mc{D},\delta_\Dl,C_\Dl)$, be an instance of $\IMP(\Dl)$ where $X=\{x_{m+1},\dots, x_{m+k}\}$, $f\in \Field[x_{m+1},\dots,x_{m+k}]$, $k=|X|$, and  $m$ will be defined later, and $\I(\cP_\Dl)\subseteq \Field[x_{m+1},\dots,x_{m+k}]$. From this we construct an instance $(f',\cP_\Gm)$ of $\IMP(\Gm)$ where $f'\in \Field[x_{1},\dots,x_{m+k}]$ and $\I(\cP_\Gm)\subseteq \Field[x_{1},\dots,x_{m+k}]$ such that $f\in \I(\cP_\Dl)$ if and only if $f'\in \I(\cP_\Gm)$.
     
     From $\cP_\Dl$ we construct an instance $\cP_{\Gm}=(\{x_{1},\dots,x_{m+k}\},\mc D,\delta_\Gm,C_\Gm)$ of $\CSP(\Gm)$ as follows.
     By the assumption each $Q\in\Dl$, say, $t_Q$-ary, is pp-definable in $\Gm$. Thus,
     \[
         Q(y_{q_Q+1},\dots,y_{q_Q+t_Q})=\exists \vc y{q_Q} \overbrace{(R_1(w^1_1,\dots, w^1_{l_1})\wedge\dots\wedge R_r(w^r_1,\dots, w^r_{l_r}))}^{L},
     \]
     where $w^1_1,\dots, w^1_{l_1},\dots,w^r_1,\dots, w^r_{l_r}\in \{\vc y{q_Q+t_Q}\}$ and $\vc Rr\sse\Gm\cup\{=_{\mc{D}}\}$. Moreover, for $\sigma$ and $\sigma_Q$, the signatures of $L$ and $Q$ respectively, we have $\sigma_Q = \sg_{|\{q_Q+1\zd q_Q+t_Q\}}$.
      Now, for every constraint $B=\langle \bs,Q\rangle\in C_\Dl$, where $\bs=(x_{i_1},\dots,x_{i_t})$ create a fresh copy of $\{\vc y{q_Q}\}$ denoted by $Y_B$, and add the following constraints to $C_\Gm$
     \[
          \langle (w^1_1,\dots, w^1_{l_1}),R_1\rangle,\dots, \langle (w^r_1,\dots, w^r_{l_r}),R_r\rangle.
     \]
     where for each $i,j$ we have $\sg_{R_i}(w_k^i)=\sg_{R_j}(w_{k'}^j)$ whenever $w_k^i$ and $w_{k'}^j$ are the same variable. In this formula, all constraints of the form $\langle(x_i, x_j),=_{\mc{D}}\rangle$  can be eliminated by replacing all occurrences of the variable $x_i$ with $x_j$.
     Set $m=\sum_{B\in C}|Y_B|$ and assume that $\cup_{B\in C}Y_B=\{\vc xm\}$.   
     
     Let $\I(\cP_\Gm)\subseteq\Field[x_{1},\dots,x_{m+k}]$ be the ideal corresponding to $\cP_{\Gm}$ and set $f'=f$. Since $f\in \Field[x_{m+1},\dots,x_{m+k}]$ we also have $f\in \Field[x_{1},\dots,x_{m+k}]$. Hence, $(f,\cP_\Gm)$ is an instance of $\IMP(\Gm)$. We prove that $f\in \I(\cP_\Dl)$ if and only if $f\in \I(\cP_\Gm)$.
     
     Suppose $f\not\in \I(\cP_\Dl)$, this means there exists $\vf\in \mb{V}(\I(\cP_\Dl))$ such that $f(\vf)\neq 0$. By Theorem~\ref{extension-theorem}, $\vf$ can be extended to a point $\vf'\in\mb{V}(\I(\cP_\Gm))$. This in turn implies that $f\not\in \I(\cP_\Gm)$. Conversely, suppose $f\not\in \I(\cP_\Gm)$. Hence, there exists $\vf'\in\mb{V}(\I(\cP_\Gm))$ such that $f(\vf')\neq 0$. The projection of $\vf'$ to its last $k$ coordinates gives a point $\vf\in \mb{V}(\I_X)$. By Lemma~\ref{variety-m-elimination}, $\vf\in \mb{V}(\I(\cP_\Dl))$ which implies $f\not\in \I(\cP_\Dl)$.
\end{proof}

\subsection{Primitive-positive interpretability}

Multisorted primitive-positive (pp-) interpretations are also similar to the one-sorted case \cite{Bulatov20:ideal}, but require a bit more care making sure that the sorts of variables match. 

\begin{definition}[pp-interpretability]
\label{pp-interpret-multi-sorted}
    Let $\Gm,\Dl$ be multi-sorted constraint languages over finite collections of sets $\mc{D}=\{D_t\mid t\in T\},\mc{E}=\{E_s\mid s\in S\}$, respectively, and $\Dl$ is finite. We say that $\Gm$ pp-interprets $\Dl$ if for every $s\in S$ there exist $i_{s,1}\zd i_{s,\ell_s}\in T$, a set $F_s \subseteq D_{i_{s,1}}\times\dots\times D_{i_{s,\ell_s}}$, and an onto mapping $\pi_s : F_s \to E_s$ such that $\Gm$ pp-defines the following relations
    \begin{enumerate}
        \item the relations $F_s$, $s\in S$,
        \item the $\pi_s$-preimage of the equality relations on $E_s$, $s\in S$, and
        \item the $\pi$-preimage of every relation in $\Dl$,
    \end{enumerate}
    where by the $\pi$-preimage of a $k$-ary relation $Q\sse E_{s_1}\tm\dots\tm E_{s_k}$ over $\mc{E}$ we mean the $m$-ary relation $\pi^{-1}(Q)$ over $\mc{D}$, with $m=\sum_{i=1}^k \ell_{s_i}$, defined by
        \[
            \pi^{-1}(Q)(x_{1,1},\ldots , x_{1,\ell_{s_1}},x_{2,1},\ldots,x_{2,\ell_{s_2}},\ldots,x_{k,1},\ldots,x_{k,\ell_{s_k}})\qquad \text{is true}
        \]
    if and only if
        \[
            Q(\pi_{s_1}(x_{1,1},\ldots , x_{1,\ell_{s_1}}),\dots,\pi_{s_k}(x_{k,1},\ldots,x_{k,\ell_{s_k}})) \qquad\text{is true}.
        \]
\end{definition}

\begin{example}
    Suppose $\mc{D}= \{\zZ_2,\zZ_3\}$ and $\mc{E}=\{\zZ_6\}$. Now, any relation on $\mc{E}$ is pp-interpretable in a language in $\mc{D}$ via $F=\zZ_2\times \zZ_3$ and $\pi:F\to \zZ_6$ as 
    \begin{table}[H]
        \centering
        \begin{tabular}{ccc}
             $\pi(0,0) = 0$ & $\pi(1,2) = 1$ & $\pi(0,1) = 2$\\
             $\pi(1,0) = 3$ & $\pi(0,2) = 4$ & $\pi(1,1) = 5$.
        \end{tabular}
        \label{tab:Z_6}
    \end{table}
\end{example}

\begin{example}\label{exa:z2-example} 
Here we present a very simple example of pp-interpretation that will be useful later. Suppose $\mc{D}=\{D=\zZ_2\}$ and $\mc{E}=\{E_1,E_2\}$ with $E_1=E_2=\zZ_2\tm\zZ_2$. Define relations $R_D=\{(0,0),(1,1)\}$ on $D$ and 
\begin{align*}
    R_E = 
    \left(
        \begin{array}{cccc}
             (0,0)& (1,0) & (0,1) & (1,1) \\
             (0,0)& (0,1) & (1,0) & (1,1)
        \end{array}
        \right)
        \begin{array}{c}
             \gets x  \\
             \gets y 
        \end{array}
\end{align*}
which is a subset of $E_1\tm E_2$, that is, it has signature $(1,2)$. Thus every entry of every tuple from $R_E$ is itself a pair like $(0,1)$. Also, the tuples from $R_E$ are written vertically as columns of the matrix, so, alternatively, $R_E$ can be represented as 
\[
R_E=\{((0,0),(0,0)),((1,0),(0,1)),((1,0),(1,0)),((1,1),(1,1))\},
\]
although the matrix representation will be more convenient in the future. 
Note that the relation $R_E$ contains all pairs $(x,y)\in E_1\tm E_2$ with $x=\begin{pmatrix} 0 & 1\\
1 & 0 \end{pmatrix}y$ (recall that each of $x,y$ is a vector from $\zZ_2\tm\zZ_2$), and the relation $R_D$ is the equality. Set $\Gm=\{R_D\}$ and $\Dl=\{R_E\}$.
    
Set $F_1=F_2=D\tm D$ and define mappings $\pi_1:F_1\to E_1$, $\pi_2:F_2\to E_2$ as follows $\pi_1(x_1,x_2)=\pi_2(x_1,x_2)=(x_1,x_2)$.
The $\pi$-preimage of the relation $R_E$ is the relation 
     \begin{align*}
        R_F = 
        \left(
        \begin{array}{cccc}
             0 & 1 & 0 & 1 \\
             0 & 0 & 1 & 1 \\
             0 & 0 & 1 & 1 \\
             0 & 1 & 0 & 1 
        \end{array}
        \right)
        \begin{array}{c}
             \gets x_1  \\
             \gets x_2 \\
             \gets y_1 \\
             \gets y_2
        \end{array}
    \end{align*}
 Here $R_F$ is a 4-ary relation, and its tuples are written as columns of this matrix.  The language $\Gm$ pp-defines $\Gm'=\{R_F\}$ through the following pp-formula 
    \[
        R_F=\{(x_1,x_2,y_1,y_2)\mid (x_1 = y_2)\land (x_2 = y_1),  \text{ and } (x_1,x_2,y_1,y_2)\in\{0,1\}^4\}.
    \]
    Consider an instance $\mc{P}=(\{x,y,z\},\mc{E},\dl,C)$ of $\MCSP(\Dl)$ where the set of constraints is 
    \[
        C=\{\langle (x,y),R_E\rangle,\langle (z,y),R_E\rangle\}
    \] 
    and $\dl$ maps $x,z$ to 1 and $y$ to 2. This basically means the requirements $x=\begin{pmatrix} 0 & 1\\
    1 & 0 \end{pmatrix}y$ and $z=\begin{pmatrix} 0 & 1\\
    1 & 0 \end{pmatrix}y$. This instance is equivalent to the following instance $\mc{P}'$ of $\CSP(\Gm')$:
    \begin{align*}
        &\langle (x_1,x_2,y_1,y_2),R_F\rangle \land 
        \langle (z_1,z_2,y_1,y_2),R_F\rangle 
    \end{align*}
     Applying the mapping $\pi$, every solution of the instance $\mc{P}'$ can be transformed to a solution of instance $\mc{P}$ and back. This in turn is equivalent to the following instance of \CSP$(\Gm)$
     \begin{align*}
        &\langle (x_1,y_2),R_D\rangle \land \langle (x_2,y_1),R_D\rangle \land \langle (z_1,y_2),R_D\rangle \land \langle (z_2,y_1),R_D\rangle. 
    \end{align*}
\end{example}

As in the one-sorted case, pp-interpretations give rise to reductions between IMPs.

\begin{theorem}
\label{thm:pp-interpret-multi-sorted}
    Let $\Gm$, $\Delta$ be multi-sorted constraint languages over collections of sets $\mc{D}=\{D_t\mid t\in T\},\mc{E}=\{E_s\mid s\in S\}$, respectively, and let $\Gm$ pp-interpret $\Delta$.
    Then $\IMP_d(\Delta)$ is polynomial time reducible to $\IMP_{O(d)}(\Gm)$.
\end{theorem}

\begin{proof}
    Recall that $\Gm$, $\Delta$ are multi-sorted constraint languages over collections of sets $\mc{D}=\{D_t\mid t\in T\},\mc{E}=\{E_s\mid s\in S\}$, respectively, and $\Gm$ pp-interprets $\Delta$.
     
     Let $(f,\cP_\Dl)$ be an instance of $\IMP_d(\Dl)$ where $f\in \Field[x_{1},\dots,x_n]$, $\cP_{\Dl}=(\{x_{1},\dots,x_n\},\mc{E},\delta_\Dl,C_\Dl)$, an instance of $\CSP(\Dl)$, and $\I(\cP_\Dl)\subseteq \Field[x_{1},\dots,x_n]$.
     
     As $\Gm$ pp-interprets $\Dl$ then for every $s\in S$ there exist $i_{s,1}\zd i_{s,\ell_s}\in T$, a set $F_s \subseteq D_{i_{s,1}}\times\dots\times D_{i_{s,\ell_s}}$, and an onto mapping $\pi_s : F_s \to E_s$ such that $\Gm$ pp-defines the constraint language $\Gm'$ that consists of the following relations:
     \begin{enumerate}
        \item the relations $F_s$, $s\in S$,
        \item the $\pi_s$-preimage of the equality relations on $E_s$, $s\in S$, and
        \item the $\pi$-preimage of every relation in $\Dl$.
    \end{enumerate}
    By Definition~\ref{pp-interpret-multi-sorted}, according to the properties of $\pi_1,\dots,\pi_{|S|}$ we can rewrite an instance of $\CSP(\Dl)$ to an instance of $\CSP(\Gamma')$.
     Note that if $\delta_{\Dl}(x)=s$ then $\delta_{\Gm'}(x_{s,j})=i_{s,j}$, for all $1\leq j\leq \ell_s$. By Theorem~\ref{thm:Multi-sorted-pp-reduction}, $\IMP(\Gamma')$ is reducible to $\IMP(\Gm)$. It remains to show $\IMP_d(\Dl)$ is reducible to $\IMP_d(\Gamma')$. To do so, from instance $(f,\cP_\Dl)$ of $\IMP_d(\Dl)$ we construct an instance $(f',\cP_{\Gm'})$ of $\IMP_d(\Gamma')$ such that $f\in\I(\cP_\Dl)$ if and only if $f'\in\I(\cP_{\Gm'})$. 
     
     Let $p_s$ be a polynomial of total degree at most $|E_s|(|D_{i_{s,1}}|+\dots+ |D_{i_{s,\ell_s}}|)$ that interpolates the mapping $\pi_s$, $s\in S$. For every $f\in \Field[x_1,\dots,x_n]$, let 
     \[
        f'\in \Field [x_{1,1},\ldots , x_{1,\ell_{s_1}},x_{2,1},\ldots,x_{2,\ell_{s_2}},\ldots,x_{n,1},\ldots,x_{n,\ell_{s_n}}]
    \]
     be the polynomial that is obtained from $f$ by replacing each indeterminate $x_i$, with $\delta_\Dl(x_i)=s$, by $p_{s}(x_{s,1},\dots,x_{s,\ell_{s}})$.
     
     Now, for any assignment $\vf:\{\vc xn\}\to \cup_{s\in S} E_s$, with the condition $\vf(x_i)\in E_{\delta_\Dl(x_i)}$ for all $x_i$, $f(\vf)=0$ if and only if $f'(\psi)=0$ for every 
     \[
        \psi:\{x_{1,1},\ldots , x_{1,\ell_{s_1}},x_{2,1},\ldots,x_{2,\ell_{s_2}},\ldots,x_{n,1},\ldots,x_{n,\ell_{s_n}}\}\to \cup_{t\in T} D_t
     \]
    such that for each $x_i$ with $\delta_\Dl(x_i)=s$ we have
    \[
        \vf(x_i)=\pi(\psi(x_{s,1}),\dots,\psi(x_{s,\ell_s})).
    \]
     Moreover, for any such $\vf,\psi$ it holds that $\vf\in \mb{V}(\I(\cP_\Dl))$ if and only if  $\psi\in \mb{V}(\I(\cP_{\Gm'}))$. This yields that 
    \[
        (\exists \vf\in \mb{V}(\I(\cP_\Dl)) \land f(\vf)\neq 0) \iff (\exists \psi\in \mb{V}((\cP_{\Gm'})) \land f'(\psi)\neq 0).
    \]
    Note that the condition that $f$ has bounded degree is important here, because otherwise $f'$ may have exponentially more monomials than $f$. This completes the proof of the theorem.
\end{proof}

Observe that the proof of Theorem~\ref{thm:pp-interpret-multi-sorted} breaks if we replace $\IMP_d(\Gm)$ with its unbounded version $\IMP(\Gm)$. Indeed, if any of the polynomials $p_s$ contains at least 2 monomials and $f$ has a monomial of degree $n$ then $f'$ obtained by substituting the $p_s$ into $f$ may have more than $2^n$ monomials, thus, the reduction would not be polynomial time.

\subsection{Polymorphisms and multi-sorted polymorphisms}

One of the standard methods to reason about constraint satisfaction problems is to use polymorphisms. Here we only give the necessary basic definitions. For more details the reader is referred to \cite{Barto17:polymorphisms,BulatovJK05}. Let $R$ be an ($n$-ary) relation on a set $D$ and $f$ a ($k$-ary) operation on the same set, that is, $f:D^k\to D$. Operation $f$ is said to be a \emph{polymorphism} of $R$, or $R$ is \emph{invariant} under $f$, if for any $\vc\ba k\in R$ the tuple $f(\vc\ba k)$ belongs to $R$, where $f$ is applied component-wise, that is, 
\[
f(\vc\ba k)=(f(a_{1,1}\zd a_{1,k})\zd f(a_{n,1}\zd a_{n,k})),
\]
and $\ba_i=(a_{1,i}\zd a_{n,i})$. The set of all polymorphisms of $R$ is denoted $\Pol(R)$. For a constraint language $\Gm$ by $\Pol(\Gm)$ we denote the set of all operations that are polymorphisms of every relation from $\Gm$. 


Polymorphisms provide a link between constraint languages and relations pp-definable in those languages.

\begin{proposition}[\cite{Bodnarchuk69:Galua1,Geiger68:closed}]\label{pro:galois}
Let $\Gm$ be a constraint language on set $A$ and $R$ a relation on the same set. The relation $R$ is pp-definable in $\Gm$ if and only if $\Pol(\Gm)\sse\Pol(R)$.
\end{proposition}

\begin{corollary}[\cite{JeavonsCG97,Bulatov20:ideal}]\label{cor:polymorphism-reduction}
Let $\Gm,\Dl$ be constraint languages on a set $D$, $\Dl$ finite, and $\Pol(\Gm)\sse\Pol(\Dl)$. Then $\CSP(\Dl)$ is polynomial time reducible to $\CSP(\Gm)$, and $\IMP(\Dl)$ [$\IMP_d(\Dl)$] is polynomial time reducible to $\IMP(\Gm)$ [$\IMP_d(\Gm)$, respectively].
\end{corollary}

We will need a version of polymorphisms adapted to multi-sorted relations. Let $\cD=\{D_t\mid t\in T\}$ be a collection of sets. A multi-sorted operation on $\cD$ is a \emph{functional symbol} $f$ with associated \emph{arity} $k$ along with an interpretation $f^{D_t}$ of $f$ on every set $D_t\in\cD$, which is a $k$-ary operation on $D_t$. A multi-sorted operation $f$ is said to be a \emph{(multi-sorted) polymorphism} of a multi-sorted relation $R\sse D_{t_1}\tm\dots\tm D_{t_n}$, $\vc tn\in T$, if for any $\vc\ba k\in R$ the tuple 
\[
f(\vc\ba k)=(f^{D_{t_1}}(a_{1,1}\zd a_{1,k})\zd f^{D_{t_n}}(a_{n,1}\zd a_{n,k}))
\]
belongs to $R$.

\begin{example}\label{exa:multisorted}
Note that for the sake of defining a multi-sorted operation, the collection $\cD$ does not have to be finite. Let $\cA$ be the class of all finite Abelian groups and $f$ a ternary functional symbol that is interpreted as the affine operation $f^\zA(x,y,z)=x-y+z$ on every $\zA\in\cA$, where $+,-$ are operations of $\zA$.  

Consider the multi-sorted binary relation $R\sse\zZ_2\tm\zZ_4$ over $\cD=\{\zZ_2,\zZ_4\}$ given by
\[
R=\{(0,1),(0,3),(1,0),(1,2)\}.
\]
It is straightforward to verify that $f$ is a polymorphism of $R$. For instance,
\[
f\left(\cl 01,\cl 10,\cl 12\right)=\cl{0-1+1}{1-0+2}=\cl 03\in R.
\]
To make sure that $f$ is a polymorphism of $R$ we of course have to check every combination of pairs from $R$.
\end{example}

The connection between multi-sorted polymorphisms and pp-definitions is more complicated than that in the one-sorted case \cite{BulatovJ03-multi-sorted}, and we do not need it here. However, we will need the following well known fact. 

\begin{lemma}
Let $R\sse D_1\tm\dots\tm D_n$, $f$ a $k$-ary polymorphism of $R$, and let $\vc fk$ be $m$-ary polymorphisms of $R$. Then the \emph{composition} of $f$ and $\vc fk$ given by
\[
g^{D_i}(\vc xm)=f^{D_i}(f_1^{D_i}(\vc xm)\zd f_k^{D_i}(\vc xm))
\]
for $i\in[n]$ is a polymorphism of $R$.
\end{lemma}

\section{CSPs over Abelian groups}
\label{sec:Abelian-csps}
CSPs over Abelian groups, or more precisely problems of the form $\CSP(\Gm)$ where $\Gm$ is a constraint language closed under the affine polymorphism $x-y+z$ of an Abelian group, are well understood. However, they are usually considered as a special case of either arbitrary finite groups, in which case the coset generation algorithm applies \cite{FV98}, or as a special case of CSPs with a Mal'tsev polymorphism \cite{Bulatov06:simple,Idziak10:tractability}. In the IMP literature \cite{Mastrolilli19,Mastrolilli21:complexity,Bharathi-Minority,Bulatov20:ideal} such CSPs have been mainly considered from the point of view of systems of linear equations. Such a representation is necessary, because it is used to construct a Gr\"obner basis of the corresponding ideal. While it is true that every CSP given by a system of linear equations over some Abelian group can also be thought of as an instance of $\CSP(\Gm)$ for an appropriate language $\Gm$ closed under the affine operation, the converse is not true in general. For instance, the relation $R_E$ from Example~\ref{exa:z2-example} is invariant under the affine operation $x-y+z$ of $\zZ_2\tm\zZ_2$, but cannot be represented by a system of linear equations over this group. 

Therefore our goal in this section is to show that a CSP over an Abelian group can always be converted into a (multi-sorted) CSP that admits a representation by a system of linear equations (with caveats that will be discussed later), and then to demonstrate how a row-echelon form of such a system can be constructed, ready to be transformed into a Gr\"obner basis. 

\subsection{Abelian groups}\label{sec:AG}

In this section we state the facts about Abelian groups we will need in this paper. 

\begin{proposition}\label{pro:FTAG}
\begin{itemize}
    \item[(1)]
    (\emph{The Fundamental Theorem of Abelian Groups.}) Let $\zA$ be a finite Abelian group. Then $\zA=\zA_1\oplus\dots\oplus\zA_n$, where $\vc\zA n$ are cyclic groups.
    \item[(2)]
    There exists a decomposition from item (1), in which the order of each $\zA_i$ is a prime power.
\end{itemize}
\end{proposition}

Let $\zA$ be an Abelian group. By Proposition~\ref{pro:FTAG}, $\zA$ can be decomposed into $\zA=\zA_1\oplus\dots\oplus\zA_n$ where $\zA_i=\zZ_{q_i^{\ell_i}}$. Without loss of generality assume that for some $\vc ks$ it holds that $q_1=\dots=q_{k_1}=p_1$, $q_{k_1+1}=\dots=q_{k_1+k_2}=p_2,\dots,$ $q_{k_1+\dots+k_{s-1}+1}=\dots=q_{k_1+\dots+k_s}=p_s$. We also change the notation for $\ell_i$ so that $\zA$ can be represented as
\[
    \zA = \zZ_{p_1^{\ell_{1,1}}}\oplus\dots \oplus\zZ_{p_1^{\ell_{1,k_1}}}\oplus \zZ_{p_2^{\ell_{2,1}}}\dots\oplus\zZ_{p_2^{\ell_{2,k_2}}} \oplus\dots \oplus \zZ_{p_s^{\ell_{s,k_s}}}.
\]
Later it will also be convenient to assume that $\ell_{r,k_r}$ is maximal among $\ell_{r,1},\dots,\ell_{r,k_r}$. We will denote this value by $m_r$.


Next we describe relations invariant under an affine polymorphism of an Abelian group in group-theoretic form, cf.\ Example~\ref{exa:multisorted}. Recall that for an Abelian group $\zA$ and its subgroup $\zB$, a \emph{coset} of $\zA$ modulo $\zB$ is a set of the form $a_0+\zB=\{a_0+a\mid a\in\zB\}$. The following statement is folklore, but we give a proof for completeness.

\begin{lemma}\label{lem:coset}
Let $R$ be a subset of the Cartesian product of Abelian groups $\zA_1\tm\dots\tm\zA_n$. Then $R$ is invariant with respect to the (multi-sorted) affine operation $f(x,y,z)=x-y+z$ of the groups $\vc\zA n$ if and only if $R$ is a coset of $\zA=\zA_1\tm\dots\tm\zA_n$, viewed as an Abelian group, modulo some subgroup $\zB$ of $\zA$.
\end{lemma}

\begin{proof}
If $R$ is a coset of $\zA$ modulo a subgroup $\zB$, fix $\ba_0\in R$. Then $R=\{\ba_0+\ba\mid \ba\in\zB\}$. For any $\ba,\bb,\bc\in\zB$ we have 
\[
f(\ba+\ba_0,\bb+\ba_0,\bc+\ba_0)=(\ba-\bb+\bc)+\ba_0\in R,
\]
as $\ba-\bb+\bc\in\zB$.

Conversely, suppose $R$ is invariant under $f$. Fix $\ba_0\in R$ and set $\zB=\{\ba-\ba_0\mid \ba\in R\}$. We need to show that $\zB$ is a subgroup of $\zA$. Since $\zA$ is finite it suffices to show that $\zB$ is closed under addition. Let $\ba,\bb\in\zB$, then $\ba+\ba_0,\bb+\ba_0\in R$. As $R$ is invariant under $f$,
\[
f(\ba+\ba_0,\ba_0,\bb+\ba_0)=(\ba+\ba_0)-\ba_0+(\bb+\ba_0)=(\ba+\bb)+\ba_0\in R,
\]
implying $\ba+\bb\in\zB$.
\end{proof}

\subsection{PP-interpretations in Abelian groups}

In this section we show that any constraint language invariant under an affine operation of some Abelian group can be pp-interpreted by a multi-sorted constraint language over very simple groups. We use the notation from Section~\ref{sec:AG}.

\begin{proposition}\label{pro:pp-interpret}
Let $\Dl$ be a finite constraint language invariant under the affine operation of $\zA$. Then there is a multi-sorted constraint language $\Gm$ over $\zZ_{p_1^{m_1}}\zd\zZ_{p_s^{m_s}}$ invariant under the affine operation of $\zZ_{p_1^{m_1}}\zd\zZ_{p_s^{m_s}}$ such that $\Gm$ pp-interprets $\Dl$.
\end{proposition}

\begin{proof}
For a natural number $r$ and an Abelian group $\zB$ by $r\zB$ we denote the subgroup $r\zB=\{ra\mid a\in \zB\}$ of $\zB$. Observing that $\zZ_{p^\ell}$ is isomorphic to $p^{m-\ell}\zZ_{p^m}$ let 
\[
F=p_1^{m_1-\ell_{1,1}}\zZ_{p_1^{m_1}}\tm\dots\tm p_1^{m_1-\ell_{1,k_1}}\zZ_{p_1^{m_1}}\tm\dots\tm p_1^{m_s-\ell_{s,k_{s-1}+1}}\zZ_{p_s^{m_s}}\tm\dots\tm p_1^{m_s-\ell_{s,k_s}}\zZ_{p_s^{m_s}}, 
\]
and define a mapping $\pi:F\to \zA$ by
\begin{align*}
    &\pi(x_{1,1},\dots,x_{1,k_1},\dots,x_{s,1},\dots,x_{s,k_s}) =\\ 
    &\qquad((p_1^{m_1-\ell_{1,1}})^{-1}x_{1,1},(p_1^{m_1-\ell_{1,2}})^{-1}x_{1,2},\dots,x_{1,k_1},\dots,(p_s^{m_s-\ell_{s,k_{s-1}+1}})^{-1}x_{s,1},\dots,(p_s^{m_s-\ell_{s,k_s}})^{-1}x_{s,k_s}).
\end{align*}
Note that the values of the form $(p_r^{m_r-\ell_{r,i}})^{-1}x_{r,i}\in\zZ_{p_r^{\ell_{r,i}}}$ are well defined because $x_{r,i}\in p_r^{m_r-\ell_{r,i}}\zZ_{p_r^{m_r}}$ and there is only one $z\in\zZ_{p_r^{\ell_{r,i}}}$ such that $x_{r,i}=p_r^{m_r-\ell_{r,i}}z$.
Then we set $$\Gm=\{F\}\cup\{\pi^{-1}(=_\zA)\}\cup\{\pi^{-1}(R)\mid R\in\Dl\}.$$ The language $\Gm$ contains $F$, the preimage of the equality relation on $\zA$, and the preimages of all the relations from $\Dl$. Therefore, by the definition of pp-interpretability $\Gm$ pp-interprets $\Dl$. It remains to show that $\Gm$ is invariant under the affine operation of $\zZ_{p_1^{m_1}}\zd\zZ_{p_s^{m_s}}$ as a multi-sorted polymorphism. 

The set $F$ is clearly invariant under any operation of the groups $\zZ_{p_1^{m_1}}\zd\zZ_{p_s^{m_s}}$, as it is a Cartesian product of subgroups of those groups. For $\ba,\bb\in F$ we have $(\ba,\bb)\in\pi^{-1}(=_\zA)$ if and only if $p_r^{m_r-\ell_{r,i}}a_{r,i}=p_r^{m_r-\ell_{r,i}}b_{r,i}$ in $\zZ_{p_r^{m_r}}$, where $i\in[k_r]$. So, if $(\ba^1,\bb^1),(\ba^2,\bb^2),(\ba^3,\bb^3)\in\pi^{-1}(=_\zA)$, $\bc=\ba^1-\ba^2+\ba^3, \bd=\bb^1-\bb^2+\bb^3$ then for any $r\in[s]$ and $i\in[k_r]$ we have 
\begin{align*}
p_r^{m_r-\ell_{r,i}}c_{r,i} &=p_r^{m_r-\ell_{r,i}}(a^1_{r,i}-a^2_{r,i}+a^3_{r,i})\\
&=p_r^{m_r-\ell_{r,i}}a^1_{r,i}-p_r^{m_r-\ell_{r,i}}a^2_{r,i}+p_r^{m_r-\ell_{r,i}}a^3_{r,i}\\
& =p_r^{m_r-\ell_{r,i}}b^1_{r,i}-p_r^{m_r-\ell_{r,i}}b^2_{r,i}+p_r^{m_r-\ell_{r,i}}b^3_{r,i}\\
&=p_r^{m_r-\ell_{r,i}}(b^1_{r,i}-b^2_{r,i}+b^3_{r,i})\\
&=p_r^{m_r-\ell_{r,i}}d_{r,i}.
\end{align*}

Now, let $R\in\Dl$ be a $t$-ary relation and $R'=\pi^{-1}(R)$. We show that $x-y+z$ is a polymorphism of $R'$. Let $\ba',\bb',\bc'\in R'$, $\ba=\pi(\ba'),\bb=\pi(\bb'),\bc=\pi(\bc')$, $\bd'=\ba'-\bb'+\bc'$, and $\bd=\ba-\bb+\bc$. It suffices to show that $\pi(\bd')=\bd$, as it implies $\bd'\in R'$, since $\bd\in R$. Each of the tuples $\ba',\bb',\bc',\bd'$ is a $t\cdot n$-tuple. We will denote its components by $a'_{i,r,j}$ ($b'_{i,r,j},c'_{i,r,j},d'_{i,r,j}$), $i\in[t],r\in[s],j\in[k_r]$ so that $\pi(a'_{i,1,1}\zd a'_{i,s,k_s})=a_i$, ($\pi(b'_{i,1,1}\zd b'_{i,s,k_s})=b_i$, $\pi(c'_{i,1,1}\zd c'_{i,s,k_s})=c_i$). For any $i\in[t]$, $r\in[s]$, and $j\in[k_r]$ we have 
\[
(p_r^{m_r-\ell_{r,j}})^{-1}d'_{i,r,j}=(p_r^{m_r-\ell_{r,j}})^{-1}a'_{i,r,j}-(p_r^{m_r-\ell_{rj}})^{-1}b'_{i,r,j}+(p_r^{m_r-\ell_{r,j}})^{-1}c'_{i,r,j}.
\]
Therefore,
\begin{align*}
    \pi(d'_{i,1,1}\zd d'_{i,s,k_s}) 
    &=\left((p_1^{m_1-\ell_{1,1}})^{-1}d'_{i,1,1},\dots ,(p_s^{m_s-\ell_{s,k_s}})^{-1}d'_{i,s,k_s}\right)\\
    &=\left((p_1^{m_1-\ell_{1,1}})^{-1}(a'_{i,1,1}-b'_{i,1,1}+c'_{i,1,1}),\dots ,(p_s^{m_s-\ell_{s,k_s}})^{-1}(a'_{i,s,k_s}-b'_{i,s,k_s}+c'_{i,s,k_s})\right)\\
    &=\left((p_1^{m_1-\ell_{1,1}})^{-1}a'_{i,1,1},\dots ,(p_s^{m_s-\ell_{s,k_s}})^{-1}a'_{i,s,k_s}\right)\\
    & \ \ \ -\left((p_1^{m_1-\ell_{1,1}})^{-1}b'_{i,1,1},\dots ,(p_s^{m_s-\ell_{s,k_s}})^{-1}b'_{i,s,k_s}\right)\\
    & \ \ \ +\left((p_1^{m_1-\ell_{1,1}})^{-1}c'_{i,1,1},\dots ,(p_s^{m_s-\ell_{s,k_s}})^{-1}c'_{i,s,k_s}\right)\\
    &=a_i-b_i+c_i=d_i.
\end{align*}
\end{proof}

A nice property of languages over $\zZ_{p_1^{m_1}}\zd\zZ_{p_s^{m_s}}$, with $p_i\neq p_j$ when $i\neq j$, is that any (multi-sorted) relation can be decomposed into relations of the same sort. Let $R\sse D_1\tm\dots\tm D_n$ and $I=\{\vc ik\}\sse[n]$. For $\ba=(a_1,\dots,a_n)\in R$ by $\pr_I\ba$ we denote the tuple $(a_{i_1}\zd a_{i_k})$, and $\pr_IR=\{\pr_I\ba\mid\ba\in R\}$.

For two relations $R_1$ and $R_2$ of arities $r$ and $t$ respectively, their Cartesian product is the relation $R$ of arity $r+t$ defined as
\begin{align*}
    R(x_1,\dots,x_r,y_1,\dots,y_t)&=R_1\times R_2\\
    &=\{(x_1,\dots,x_r,y_1,\dots,y_t)\mid (x_1,\dots,x_r)\in R_1 \land (y_1,\dots,y_t)\in R_2\}.
\end{align*}

\begin{lemma}\label{lm:decomposition}
    Let $R(x_{1,1},\dots,x_{1,k_1},\dots,x_{s,1},\dots,x_{s,k_s})$ be such that $x_{i,j}$ has domain $\zZ_{p_i^{m_i}}$ and $R$ is invariant under the affine operation. Then $R$ is decomposable as follows
    \begin{align*}
     &R(x_{1,1},\dots,x_{1,k_1},\dots,x_{s,1},\dots,x_{s,k_s})=\\
     &\qquad\qquad\qquad
     (\pr_{(1,1),\dots,(1,k_1)}R)(x_{1,1},\dots,x_{1,k_1})\times\dots\times (\pr_{(s,1),\dots,(s,k_s)}R)(x_{s,1},\dots,x_{s,k_s}).
    \end{align*}
\end{lemma}

\begin{proof}
It suffices to show that 
\begin{align*}
     &R(x_{1,1},\dots,x_{1,k_1},\dots,x_{s,1},\dots,x_{s,k_s})=\\
     &\qquad\qquad\qquad
     (\pr_{(1,1),\dots,(1,k_1)}R)(x_{1,1},\dots,x_{1,k_1})\times (\pr_{(2,1),\dots,(s,k_s)}R)(x_{2,1},\dots,x_{s,k_s}).
    \end{align*}
    
Let $M=p_1^{m_1}\cdot p_2^{m_2}\ldots\cdot p_s^{m_s}$, $M_1=M/p_1^{m_1}$, and let $u,v\in[M]$ be such that $u\equiv 1\pmod{p_1^{m_1}},u\equiv 0\pmod{M_1}$ and $v\equiv 0\pmod{p_1^{m_1}}, v\equiv 1\pmod{M_1}$. It is possible, as all the $p_i$'s are different primes. Then for any $\ba\in R$, $\ba=(\ba_1,\ba_2)$, where $\ba_1\in\pr_{(1,1),\dots,(1,k_1)}R$ and $\ba_2\in\pr_{(2,1),\dots,(s,k_s)}R$, it holds that $u\cdot\ba=(\ba_1,\ov0)$ and $v\cdot\ba=(\ov0,\ba_2)$, and $\ov0$ denotes the zero vector of an appropriate length. Note that $u+v=1\pmod M$.

We prove that, as any composition of polymorphisms of $R$ is a polymorphism of $R$ the operation $g(x,y,z)=ux+vy+(1-u-v)z$ is a polymorphism of $R$. Note that this argument does not depend on the above properties of $u$ and $v$. More precisely, we prove by induction on $u',v'$ that $u'x+v'y+(1-u'-v')z$ can be obtained as a composition of $f(x,y,z)=x-y+z$ with itself. Indeed, for $u'=v'=1$ the operation $f(x,z,y)=x+y-z$ is as required. Suppose the statement is proved for $g'(x,y,z)=u'x+v'y+(1-u'-v')z$. Then
\begin{align*}
& f(x,z,g'(x,y,z))=x-z+u'x+v'y+(1-u'-v')z=(u'+1)x+v'y+(1-(u'+1)-v')z,\\
& f(y,z,g'(x,y,z))=y-z+u'x+v'y+(1-u'-v')z=u'x+(v'+1)y+(1-u'-(v'+1))z.
\end{align*}

We need to prove that if $\ba,\bb\in R$, $\ba=(\ba_1,\ba_2),\bb=(\bb_1,\bb_2)$ then $(\ba_1,\bb_2)\in R$.
This is however straightforward:
\begin{align*}
g(\ba,\bb,\ba) &=(1-v)\ba+v\bb\\
&=(1-v)(\ba_1,\ba_2)+v(\bb_1,\bb_2)\\
&=((1-v)\ba_1,(1-v)\ba_2)+(v\bb_1,v\bb_2)\\
&=(\ba_1,\bar 0)+(\bar 0,\bb_2) \\
&=(\ba_1,\bb_2).
\end{align*}
\end{proof}

Lemma~\ref{lm:decomposition} allows us to decompose multi-sorted \CSP s into instances, in which every constraint contains variables of only one sort. Let $\Gm$ is a multi-sorted constraint language over $\cD=\{\zZ_{p_1^{m_1}}\zd\zZ_{p_s^{m_s}}\}$. For every (say, $n$-ary) relation $R\in\Gm$ with signature $(i_1,\dots,i_n)$, let $J_R(u)=\{j\mid i_j=u\}$, $u\in[s]$, and $R^u=\pr_{J_R(u)}R$. Note that for every $R\in\Gm$ and every $u\in[s]$ the relation $R^u$ is a relation over a single domain $\zZ_{p_u^{m_u}}$. Set $\Gm'=\{R^u\mid R\in\Gm,u\in[s]\}$. 

\begin{proposition}\label{pro:csp-decomposition}
Let $\cP$ be an instance of $\CSP(\Gm)$, where $\Gm$ is a multi-sorted constraint language over $\cD=\{\zZ_{p_1^{m_1}}\zd\zZ_{p_s^{m_s}}\}$ invariant with respect to the affine polymorphism of $\zZ_{p_1^{m_1}}\zd\zZ_{p_s^{m_s}}$. Then there is an instance $\cP'$ of $\MCSP(\Gm')$ such that the set of variables $X$ of $\cP'$ is the same as that of $\cP$ and for any $x\in X$ its sort is the same in both $\cP$ and $\cP'$, and $\Sol(\cP)=\Sol(\cP')$. 
\end{proposition}

\begin{proof}
The instance $\cP'$ is constructed as follows: For every constraint $\ang{\bs,R}$, $\bs=(x_{i_1},\dots,x_{i_n})$ of $\cP$, we introduce constraints $\ang{\bs^u,R^u}$, $u\in[s]$, where $\bs^u=(x_{i_j})_{j\in J_R(u)}$. The result then follows by construction and Lemma~\ref{lm:decomposition}.
\end{proof}


\subsection{Constructing a system of linear equations}

Having the decomposition result in \Cref{lm:decomposition} and Proposition~\ref{pro:csp-decomposition}, we proceed to show that any instance of $\CSP(\Gm)$ where $\Gm$ is invariant under the affine operation of $\zZ_{{p}^{m}}$ can be transformed into a system of linear equations over $\zZ_{{p}^{m}}$ that is in the \emph{reduced row-echelon} form i.e., there are free variables and the rest of variables are linear combinations thereof. Note that unlike linear equations over $\zZ_{p}$, such a transformation is not immediate. In particular, it will require introducing new variables that will serve as free variables. Recall that $\Sol(\cP)$ denotes the set of solutions of $\cP$.

\begin{lemma}\label{lem:to-equations}
Let $\cP$ be an instance of $\CSP(\Gm)$ where $\Gm$ is a constraint language over $\zA=\zZ_{{p}^{m}}$ invariant under the affine operation of the group. Let $X=\{\vc xn\}$ be the set of variables of $\cP$. Then there are variables $\vc yr$ such that for every $j\in[n]$ there are coefficients $\al_{1,j}\zd\al_{r,j},c_j\in\zZ_{p^m}$, for which $(\vc xn)\in\Sol(\cP)$ if and only if $x_j+\al_{1,j}y_1+\dots+\al_{r,j}y_r+c_j=0$ for some values of $\vc yr$ from $\zZ_{p^m}$.
\end{lemma}

\begin{proof}
We start with a claim that indicates what the existing algorithms allow us to do with respect to the instance $\cP$. 

\medskip

\noindent
{\sc Claim 1.}
(1) A solution of $\cP$, if one exists, can be found in polynomial time.\\[2mm]
(2) For any $x\in X$ and any $a\in\zZ_{p^m}$, a solution $\vf\in\Sol(\cP)$ can be found in polynomial time such that $\vf(x)=a$, if one exists.\\[2mm]
(3) For any $x\in X$ the set $\Sol_x(\cP)=\{\vf(x)\mid \vf\in\Sol(\cP)\}$ can be found in polynomial time.

\begin{proof}[Proof of Claim 1]
(1) A ternary operation $f$ on a set $A$ is said to be \emph{Mal'tsev} if $f(x,x,y)=f(y,x,x)=y$ for $x,y\in A$. The affine operation of any Abelian group including $\zA$ is Mal'tsev. It was proved in \cite{Bulatov06:simple} that for any $\Gm$ invariant under a Mal'tsev operation the problem $\CSP(\Gm)$ can be solved in polynomial time. Since $\Gm$ in \Cref{lem:to-equations} is invariant under a Mal'tsev operation, it implies item~(1).

(2) The constant relation $R_a=\{(a)\}$ is invariant under the affine operation of $\zA$. This means that the problem $\cP'$ obtained from $\cP$ by adding the constraint $\ang{(x),R_a}$ can be solved in polynomial time using the algorithm from \cite{Bulatov06:simple}. A mapping $\vf$ is a solution of $\cP'$ if and only if $\vf\in\Sol(\cP)$ and $\vf(x)=a$.

(3) To find the set $\Sol_x(\cP)$ one just needs to apply item (2) to every element of $\zA$.
\end{proof}

Note that, by Proposition \ref{pro:galois}, $\Sol(\cP)$ is invariant under the polymorphisms of the language $\Gamma$. Pick an arbitrary solution $\vf_0\in\Sol(\cP)$ and let $S'=\{\vf-\vf_0\mid \vf\in\Sol(\cP)\}$. By Lemma~\ref{lem:coset} $S'$ is a subgroup of $\zA^n$. 
First, find $\Sol_x(\cP)$ for every $x\in X$ and set $S'_x=\Sol_x(\cP)-\vf_0(x)$ (subtraction in $\zA$). For $a\in\zA$ let $p(a)$ denote the maximal power of $p$ that divides $a$. Find $x\in X$ such that $S'_x$ contains an element $a\in\zA$ for which $p(a)$ is minimal possible. Without loss of generality let $x$ be $x_1$, and denote the value $a$ by $a_1$ and set $o_1=m-p(a_1)$. Let also $\vf_1$ be a solution of $\cP$ such that $\vf_1(x_1)=a_1+\vf_0(x_1)$, and $\vf'_1=\vf_1-\vf_0$.
Observe that since $p(a_1)$ is minimal, for any $x_i\in X$ we have $\vf'_1(x_i)=\al'_{1,i}\vf'_1(x_1)$ for some $\al'_{1,i}\in\zZ_{p^m}$.
Let $\cP^{(1)}$ be the instance $\cP$ with the extra constraint $\ang{(x_1),R_{\vf_0(x_1)}}$. 
Suppose that $\cP^{(i)}$ is constructed, that is obtained from $\cP$ by adding constraints $\ang{(x_j),R_{\vf_0(x_j)}}$ for $j\in[i]$. Let also $S^{(i)}=\{\vf-\vf_0\mid \vf\in\Sol(\cP^{(i)})\}$. Again, find $x\in X-\{\vc xi\}$ and an element $a\in S^{(i)}_x$ such that $p(a)$ is minimal possible. Assume that $x=x_{i+1}$, $a=a_{i+1}$, and $o_{i+1}=m-p(a_{i+1})$. Find a solution $\vf_{i+1}\in\Sol(\cP^{(i)})$ with $\vf_{i+1}(x_{i+1})=a_{i+1}+\vf_0(x_{i+1})$ and let $\vf'_{i+1}=\vf_{i+1}-\vf_0$.

The process ends at some point, suppose at step $r$, as $\vf_0$ is the only solution of $\cP^{(r+1)}$. By construction, $\{\vf'_1\zd\vf'_r\}$ is a generating set of the group $S'$ and $\vf'_j(x_i)=0$ for $j>i$, $i,j\in[r]$. Indeed, for the latter statement, if $j>i$ then for any $\vf\in\Sol(\cP^{(j)})$, it holds that $\vf(x_i)=\vf_0(x_i)$, and therefore $\vf'_j(x_i)=\vf_j(x_i)-\vf_0(x_i)=0$. For the former statement, that by the choice of $a_i$ --- the minimality of $p(a_i)$ --- for any $\vf\in S^{(i-1)}$ there is $\beta\in\zZ_{p^{o(i)}}$ such that $\vf(x_i)=\beta\vf'_i(x_i)$, and therefore $S^{(i-1)}$ is generated by the elements of $S^{(i)}$ and $\vf'_i$.

As we observed, for any $j\in[n]$ there are $\al'_{1,j}\zd\al'_{r,j}$, $\al'_{i,j}\in\zZ_{p^{o_i}}$ for $i\in[r]$, such that $\vf'_1(x_j)=\al'_{1,j}\vf'_1(x_1)\zd \vf'_r(x_j)=\al'_{r,j}\vf'_r(x_r)$. We claim that coefficients $\al_{i,j}=-\al'_{i,j}\vf'_i(x_i)$, $c_j=-\vf_0(x_i)$ $i\in[r],j\in[n]$ are as required. To see this, observe that, as $\{\vf'_1\zd\vf'_r\}$ is a generating set of $S'$, we have $\vf\in\Sol(\cP)$ if and only if there are $\vc yr\in\zZ_{p^m}$, such that $\vf-\vf_0=y_1\vf'_1+\dots+y_r\vf'_r$. Thus, for any $j\in[n]$ we have 
\begin{align*}
\vf(x_j)+c_j &=y_1\vf'_1(x_j)+\dots+y_r\vf'_r(x_j)\\
&=y_1\al'_{1,j}\vf'_1(x_1)+\dots+y_r\al'_{r,j}\vf'_r(x_r)\\
&=-\al_{1,j}y_1-\dots-\al_{r,j}y_r.
\end{align*}
Thus, solutions of $\cP$ are exactly the mappings satisfying the equations 
\[
x_j+\al_{1,j}y_1+\dots+\al_{r,j}y_r+c_j=0.
\]
\end{proof}

Putting everything together, we have proposed a reduction that transforms every instance of $\CSP(\Dl)$, where $\Dl$ is a constraint language invariant under the affine operation of an Abelian group, into a systems of linear equations over cyclic $p$-groups. Note that by the Decomposition Lemma (Lemma~\ref{lm:decomposition}) and Proposition~\ref{pro:csp-decomposition} we can assume that these systems of linear equations do not share variables. We summarize the result of this section as the following proposition.

\begin{proposition}
\label{pro:CSP-to-LIN}
Let $\Dl$ be a constraint language invariant under the affine operation of an Abelian group $\zA$. There are distinct primes $p_1,\dots,p_s$, integers $m_1,\dots,m_s$ (not necessary distinct), and a multi-sorted constraint language $\Gm$ over $\zZ_{p_1^{m_1}},\dots,\zZ_{p_s^{m_s}}$ such that $\Gm$ is invariant under the affine operation of these groups, and $\Gm$ pp-interprets $\Dl$. Moreover, for every instance $\cP$ of $\CSP(\Gm)$ 
there are integers $k_1,\dots,k_s$ (not necessary distinct) such that $\cP$ is on the set of variables $X=\{x_{1,1}\zd x_{1,k_1}\zd x_{s,1}\zd x_{s,k_s}\}$, and it can be expressed as $s$ systems $\mc{L}_1,\dots,\mc{L}_s$ of linear equations where 

\begin{itemize}
 \item[1.] each $\mc{L}_i$ is a system of linear equations over $\zZ_{p_i^{m_i}}$ with variables $X(\mc L_i)\cup Y(\mc L_i)$, where $X(\mc{L}_i)=\{x_{i,1},\dots,x_{i,k_i}\}$, $Y(\mc L_i)=\{y_{i,1}\zd y_{i,r_i}\}$;
 \item[2.] each $\mc{L}_i$ is of the following form
\begin{align*}
(\mathbb{1}_{k_i\times k_i}~~ B_i)(x_{i,1},\dots,x_{i,k_i},y_{i,1},\dots,y_{i,r_i},1)^T=\mathbf{0},
\end{align*}
where $\mathbb{1}_{k_i\times k_i}$ denotes the $k_i$ by $k_i$ identity matrix, $B_i$ is a $k_i$ by $r_i+1$ matrix over $\zZ_{p_i^{m_i}}$,
\item[3.] $X(\mc{L}_i)\cap X(\mc{L}_j)=\emptyset$, $Y(\mc{L}_i)\cap Y(\mc{L}_j)=\emptyset$, for all $1\le i,j\leq s$ and $i\neq j$; 
\item[4.] an assignment $\vf$ to variables from $X$ is a solution of $\cP$ if and only if for every $i\in[s]$ there are values of variables from $Y(\mc L_i)$ that together with $\vf_{|X(\mc L_i)}$ satisfy $\mc L_i$.
\end{itemize}
\end{proposition}   
\section{Solving the IMP}
\label{sec:solving-imp}
In this section we focus on solving the \IMP\ for constraint languages that are invariant under the affine operation of a finite Abelian group. In fact, in Section \ref{sec:redu-by-sub}, we will prove that in such cases one can efficiently construct a $d$-truncated \GB\ with respect to a \grlexns. In this section we focus on decidability of $\IMP_d(\Dl)$, that is, the first part of Theorem~\ref{the:main-intro}. Formally, 

\begin{theorem}
\label{thm:main-decision}
Let $\zA$ be a finite Abelian group and $\Dl$ a constraint language such that the affine operation $x-y+z$ of $\zA$ is a polymorphism of $\Dl$. Then $\IMP_d(\Dl)$ can be solved in polynomial time for any $d$.
\end{theorem}

The rest of this section is devoted to proving Theorem~\ref{thm:main-decision}. We use the notation from Section~\ref{sec:AG}. Let $\zA$ be a finite Abelian group and $\Dl$ a finite constraint language invariant under the affine operation of $\zA$. We will provide a polynomial time algorithm that, for any instance $\mc{P}=(X,\zA,\mc{C})$ of $\CSP(\Dl)$, decides if an input polynomial $f\in \Complex[X]$ belongs to $\I(\mc{P})\subseteq \Complex[X]$. As before, we assume
\[
    \zA = \zZ_{p_1^{\ell_{1,1}}}\oplus\dots \oplus\zZ_{p_1^{\ell_{1,k_1}}}\oplus \zZ_{p_2^{\ell_{2,1}}}\dots\oplus\zZ_{p_2^{\ell_{2,k_2}}} \oplus\dots \oplus \zZ_{p_s^{\ell_{s,k_s}}},
\]
and $m_r$ is maximal among $\ell_{r,1}\zd\ell_{r,k_r}$. By \Cref{pro:pp-interpret}, $\Dl$ is pp-interpretable in a multi-sorted constraint language $\Gm$ over $\zZ_{p_1^{m_1}},\dots,\zZ_{p_s^{m_s}}$, which is invariant under the affine operation of these groups. By Theorem~\ref{thm:pp-interpret-multi-sorted}, since the multi-sorted constraint language $\Gm$ pp-interprets $\Dl$ then $\IMP_d(\Dl)$ is polynomial time reducible to $\IMP_{O(d)}(\Gm)$. Combined with \Cref{pro:CSP-to-LIN} this yields the following statement.

\begin{proposition}
\label{pro:IMP-LIN}
Let $\Dl$ be a constraint language that is invariant under the affine operation of $\zA$. Then $\IMP_d(\Dl)$ is polynomial time reducible to $\IMP_{O(d)}(\Gm)$ with $\Gm$ being a constraint language invariant under the affine operation of $\zZ_{p_1^{m_1}},\dots,\zZ_{p_s^{m_s}}$. Moreover, every instance $(f,\cP)$ of $\IMP_d(\Dl)$ is transformed to an instance $(f',\cP')$ of $\IMP_{O(d)}(\Gm)$ satisfying the following conditions. 
\begin{itemize}
\item[(1)] 
For every $i\in[s]$ there is a set $Y_i=\{y_{i,1}\zd y_{i,r_i}\}$ of variables of $\cP'$ and $Y_i\cap Y_j=\emptyset$ for $i\ne j$.
\item[(2)]
For every constraint $\ang{\bs,R}$ of $\cP'$ the following conditions hold:
    \begin{itemize}
        \item[(a)] 
        there is $i\in[s]$ such that $\zZ_{p_i^{m_i}}$ is the domain of every variable from $\bs$;
        \item[(b)] 
        $R$ is represented by a linear equation of the form
        \[
        x_j+\al_1y_{i,1}+\dots+\al_{r_i}y_{i,r_i}+\alpha_j=0
        \]
        over $\zZ_{p_i^{m_i}}$. 
    \end{itemize}
\end{itemize}
\end{proposition}

For an instance $(f,\cP)$ of $\IMP_d(\Gm)$ we assume that 
\[
    \mc{P}=(X\cup Y, \mc{D},\delta,\mc{C}),
\] 
with $X=\{x_{1,1},\dots,x_{1,k_1},\dots,x_{s,1},\dots,x_{s,k_s}\}$, $Y=\{y_{1,1},\dots,y_{1,r_1},\dots,y_{s,1},\dots,y_{s,r_s}\}$ $\mc{D}=\{\mc{D}_i\mid \mc{D}_i=\zZ_{p_i^{m_i}},1\le i\le s\}$ and $\delta:X\cup Y\to [s]$ defined as $ \delta(x_{i,j})=\dl(y_{i,j'})=i$. Furthermore, the input polynomial $f$ is from $\Complex[x_{1,1},\dots,x_{1,k_1},\dots,x_{s,1},\dots,x_{s,k_s}]$. 

In the next sections, we present a reduction that transforms the problem to an equivalent problem over roots of unities and then computes \GB\ with respect to a \lex order. 

\subsection{Reduction to roots of unity}
By Propositions~\ref{pro:CSP-to-LIN} and \ref{pro:IMP-LIN} any instance of $\CSP(\Gm)$ can be thought of as a system of linear equations.


Note that a system of linear equations over $\zZ_{p_i^{m_i}}$ can be solved in polynomial time. This immediately tells us if $1 \in \I(\cP)$ or not, and we proceed only when $1 \not\in \I(\cP)$. We assume the lexicographic order $\succ_\lex$ with 
\begin{align}
\label{lex-order}
    & x_{1,1}\succ_\lex\dots\succ_\lex x_{1,k_1}\succ_\lex\dots\succ_\lex x_{s,1}\succ_\lex\dots\succ_\lex x_{s,k_s} \\
    & \qquad\qquad\succ y_{1,1}\succ\dots\succ y_{1,r_1}\succ y_{2,1}\succ\dots\succ y_{2,r_2}\succ\dots\succ y_{s,r_s} \nonumber.
\end{align}

Since these systems of linear equations do not share any variables we construct a \GB\ for each of them independently, and then will show that the union of all these \GBs\ is indeed a \GB\ for $\I(\mc{P})$, with respect to the \lex order in \eqref{lex-order}. We denote the corresponding ideal for each $\mc{L}_i$ by $\I(\mc{L}_i)$.

Note that each linear system $\mc{L}_i$ is already in its reduced row-echelon form with $x_{i,j}$ as the leading monomial of the $j$-th equation, $1\leq j\leq k_i$. Each linear equation can be written as $x_{i,j}+f_{i,j}=0 \pmod {p_i^{m_i}}$ where $f_{i,j}$ is a linear polynomial over $\zZ_{p_i^{m_i}}$. This is elaborated on as follows.


\resizebox{\linewidth}{!}{
 \begin{minipage}{\textwidth}
\begin{align*}
\mc{L}_i : =\left\{
 \begin{array}{cccc}
      & x_{i,1} + \overbrace{\alpha_{1,1}~y_{i,1}+\dots + \alpha_{1,r_i}~y_{i,r_i}+\alpha_1}^{f_{i,1}} & = & 0  \pmod {p_i^{m_i}}\\
      & x_{i,2} + \overbrace{\alpha_{2,1}~y_{i,1}+\dots + \alpha_{2,r_i}~y_{i,r_i}+\alpha_2}^{f_{i,2}} & = & 0 \pmod {p_i^{m_i}}\\
      & & \vdots &
      \\
      & x_{i,k_i} + \overbrace{\alpha_{k_i,1}~y_{i,1}+\dots + \alpha_{k_i,r_i}~y_{i,r_i}+\alpha_{r_i}}^{f_{i,k_i}} & = & 0 \pmod {p_i^{m_i}}
 \end{array}
 \right.
 \longrightarrow
 \mc{L}_i : =\left\{
 \begin{array}{cccc}
      & x_{i,1} + f_{i,1} & = & 0  \pmod {p_i^{m_i}}\\
      & x_{i,2} + f_{i,2} & = & 0 \pmod {p_i^{m_i}}\\
      & & \vdots &
      \\
      & x_{i,k_i} + f_{i,k_i} & = & 0 \pmod {p_i^{m_i}}
 \end{array}
 \right.
\end{align*}%
\end{minipage}
}

 Hence, we can write down a generating set for $\I(\mc{L}_i)$ in an implicit form as follows where the addition is modulo $\zZ_{p_i^{m_i}}$,
    \begin{align}
    \label{eq:L_i}
        G_i=\left\{ x_{i,1} + f_{i,1},\dots, x_{i,k_i}+ f_{i,k_i} ,\prod_{j\in\zZ_{p_i^{m_i}}}(y_{i,1}-j),\dots,\prod_{j\in\zZ_{p_i^{m_i}}}(y_{i,r_i}-j)\right\}
    \end{align} 

 Let $U_{p_i^{m_i}}=\{\om_i,\om_i^2,\dots,\om_i^{(p_i^{m_i})}=\om_i^0=1\}$ be the set of $p_i^{m_i}$-th roots of unity where $\om_i$ is a primitive $p_i^{m_i}$-th root of unity. For a primitive $p_i^{m_i}$-th root of unity $\om_i$ we have $\om_i^a=\om_i^b$ if and only if $a \equiv b ~(\mathrm{mod}~p_i^{m_i})$. From $\mc{L}_i$ we construct a new CSP instance $\mc{L}'_i= (V,U_{p_i^{m_i}},\widetilde{C})$ where for each equation $x_{i,t}+ f_{i,t}=0 \pmod {p_i^{m_i}}$ we add the constraint $x_{i,t}-f_{i,t}'=0$ with 
 \[
    f_{i,t}'=\om_i^{\alpha_t}\cdot\left(y_{i,1}^{\alpha_{t,1}}\cdot\ldots \cdot y_{i,r_i}^{\alpha_{t,r_i}}\right).
 \]
    Moreover, the domain constraints are different. For each variable $x_{i,j}$, $j\in[k_i]$, or $y_{i,j}$, $j\in[r_i]$ the domain polynomial is $(x_{i,j})^{(p_i^{m_i})}=1$, $(y_{i,j})^{(p_i^{m_i})}=1$. Therefore, we represent $G_i$ from \eqref{eq:L_i} over complex number domain as follows.

    \begin{align}
        \label{eq:G1-complex}
            G_i'=\left\{ x_{i,1} - f_{i,1}',\dots, x_{i,k_i}-f_{i,k_i}' ,(y_{i,1})^{(p_i^{m_i})}-1,\dots,(y_{i,r_i})^{(p_i^{m_i})}-1\right\}
    \end{align} 

 Define univariate polynomial $\phi_i\in \zC[X]$ so that it interpolates points 
 \[
    (0,\om_i^0),(1,\om_i),\dots,(p_i^{m_i}-1,\om_i^{(p_i^{m_i}-1)}).
\]
This polynomial provides a one-to-one mapping between solutions of $\mc{L}_i$ and those of $\mc{L}_i'$. That is, $(a_{i,1},\dots,$ $a_{i,k_i},b_{i,1},\dots,b_{i,r_i})$ is a solution of $\mc{L}_i$ if and only if $(\phi_i(a_{i,1}),\dots,\phi_i(a_{i,k_i}),\phi_i(b_{i,1}),\dots,\phi_i(b_{i,r_i}))$ is a solution of $\mc{L}_i'$.

For an instance $\mc{P}$ of $\CSP(\Gm)$, which is a collection of systems of linear equations $\mc{L}_1,\dots,\mc{L}_s$, define the instance $\mc{P}'$ which is a collection of systems of linear equations $\mc{L}'_1,\dots,\mc{L}'_s$. In the next lemma we prove our transformation to roots of unity gives rise to an equivalent ideal membership problem.

\begin{lemma}
    \label{lem:equivalent}
    For a polynomial $p\in\Complex[X]$ define polynomial $p'\in\Complex[X]$ to be
    \begin{align*}
        & p'(x_{1,1},\dots,x_{1,k_1},\dots,x_{s,1},\dots,x_{s,k_s}) \\
        &\qquad\qquad = p\left(\phi_1^{-1}(x_{1,1}),\dots,\phi_1^{-1}(x_{1,k_1}),\dots,\phi_s^{-1}(x_{s,1}),\dots,\phi_s^{-1}(x_{s,k_s})\right).
    \end{align*}
    Then $p\in\I(\mc{P})$ if and only if $p'\in \I(\mc{P}')$.
\end{lemma}
\begin{proof}
    Recall that $\I(\mc{P})$ is a radical ideal, then by the Strong Nullstellensatz we have $p\not\in \I(\mc{P})$ if and only if there exists a point 
    \[
        \mb{a}\in \zZ_{p_1^{m_1}}^{k_1}\tm\dots\tm \zZ_{p_s^{m_s}}^{k_s} 
    \] 
    such that $\mb{a}$ is in $\Sol(\mc{P})$ and $p(\mb{a})\neq 0$. Similarly, as $\I(\mc{P}')$ is radical \footnote{This is because for all $i\in [s]$ and for all $j\in\{1,\dots,k_i\}$, the domain polynomial $(x_{i,j})^{(p_i^{m_i})}-1$ belongs to the ideal $\I(\mc{P}')$ i.e.\ the remainder of division of $(x_{i,j})^{(p_i^{m_i})}-1$ by $(x_{i,j})-f'_{i,j}$ is $1-(f'_{i,j})^{(p_i^{m_i})}=0$. } then $p'\not\in \I(\mc{P}')$ if and only if there exists a point $(\ba',\ba'')$,
    \[
        \mb{a}'\in U_{p_1^{m_1}}^{k_1} \times\dots\times U_{p_s^{m_s}}^{k_s},\qquad \ba''\in U_{p_1^{m_1}}^{r_1} \times\dots\times U_{p_s^{m_s}}^{r_s}
    \] 
    such that $(\mb{a}',\ba'')\in \Sol(\mc{P}')$ and  $p'(\mb{a}')\neq 0$.
   
    Moreover, by our construction, $\mb{a}=(a_{1,1},\dots,a_{1,k_1},a_{2,1},\dots,a_{2,k_2},\dots,a_{s,k_s})$ is a solution of $\mc{P}$ if and only if 
    \[
        \mb{a}'= \left(\phi_1(a_{1,1}),\dots,\phi_1(a_{1,k_1}),\phi_2(a_{2,1}),\dots,\phi_2(a_{2,k_2}),\dots,\phi_s(a_{s,k_s})\right)
    \]
    can be extended to a solution of $\mc{P}'$. Finally, 
    \[
        p'\left(\phi_1^{-1}(a_{1,1}),\dots,\phi_1^{-1}(a_{1,k_1}),\phi_2^{-1}(a_{2,1}),\dots,\phi_2^{-1}(a_{2,k_2}),\dots,\phi_s^{-1}(a_{s,k_s})\right) =  0
    \]
    if and only if $p(\mb{a})=0$. This finishes the proof.
\end{proof}

    

\subsection{Constructing \GBs}
   Having transformed the problem to a problem over (multi-sorted) roots of unity has a huge advantage, namely these new generating sets corresponding to each $\mc{L}_i$ are indeed \GBs. We first verify this for each $\mc{L}_i$, then we will show the union of all these generating sets gives a \GB\ for the entire problem. 
\begin{lemma}
    \label{lem:GB-P'}
    For each $1\le i\le s$, the set of polynomials $G_i'$ in \eqref{eq:G1-complex} is a \GB\ for $\I(\mc{L}_i')=\mb{I}(\Sol(\mc{L}_i'))$ with respect to \lex order $x_{i,1}\succ\dots\succ x_{i,k_i}\succ y_{i,1}\succ\dots\succ y_{i,r_i}$.
\end{lemma}    
 \begin{proof}
  The proof has two parts. In the first part we show that $G'_i$ is a \GB\ by showing that it satisfies the Buchberger's Criterion, Theorem~\ref{th:crit}. In the second part we show the ideal generated by $G_i'$ is equivalent to the vanishing ideal of $\Sol(\mc{L}_i')$. 
 
 Consider $\Ideal{G_i'}$. We show that $G_i'$ is a \GB\ for $\Ideal{G_i'}$ by verifying that the leading monomials of every pair of polynomials in $G_i'$ are relatively prime. For each $x_{i,j} - f_{i,j}'$, we have $\LM(x_{i,j} - f_{i,j}') = x_{i,j}$. Moreover, the leading monomial of $(y_{i,j})^{(p_i^{m_i})}-1$ is $(y_{i,j})^{(p_i^{m_i})}$.  Hence, for every pair of polynomials in $G_i'$ the leading monomials are relatively prime which, by Proposition~\ref{prop:prime-LM}, implies their reduced S-polynomial is zero. By Buchberger's Criterion, Theorem~\ref{th:crit}, it follows that $G'_i$ is a \GB\ for $\Ideal{G_i'}$ (according to the \lex order).
 
 It remains to show $\Ideal{G_i'}= \mb{I}(\Sol(\mc{L}_i'))$. It is easy to see that $\Ideal{G_i'}\subseteq \mb{I}(\Sol(\mc{L}_i'))$. This is because, by our construction we have $\Variety{\Ideal{ G_i'}}= \Sol(\mc{L}_i')$ and hence any polynomial $p\in \Ideal{ G_i'}$ is zero over all the points in $\Sol(\mc{L}_i')$ which implies $p\in \mb{I}(\Sol(\mc{L}_i'))$. Next we prove $\mb{I}(\Sol(\mc{L}_i'))\subseteq \Ideal{G_i'}$. Consider a polynomial $p\in \mb{I}(\Sol(\mc{L}_i'))$. We prove $p|_{G_i'} = 0$ i.e., normal form of $p$ by $G_i'$ is zero. Note that $p(\mb{a})=0$ for all $\mb{a}\in \Sol(\mc{L}_i')$. Let $q = p|_{G_i'}$. The polynomial $q$ does not contain variables $x_{i,1},\dots,x_{i,k_i}$. This is because $q=p|_{G_i'}$ is derived by dividing $p$ by $G_i'$, which eliminates any occurrence of $x_{i,1},\dots,x_{i,k_i}$. Specifically, these variables are the leading monomials of $x_{i,1}-f'_{i,1},..., x_{i,k_i}-f'_{i,k_i}$, thus they are canceled out during the division process. Consequently, $q$ is a polynomial in $y_{i,1},\dots,y_{i,r_i}$. Now note that any $\mb{b}=(b_{1},\dots,b_{r_i})\in U^{r_i}_{p_i^{m_i}}$ extends to a unique point in $\Sol(\mc{L}_i')$, this is because in $G_i'$ (similarly in $G_i$ and $\mc{L}_i$) all the $x_{i,1},\dots,x_{i,k_i}$ have coefficients and exponent equal to $1$.
 Therefore, all the points in $U^{r_i}_{p_i^{m_i}}$ are zeros of $q$, hence $q=p|_{G_i'}$ is the zero polynomial. Since $p\in \mb{I}(\Sol(\mc{L}_i'))$ was arbitrary chosen, it follows that for every $p\in \mb{I}(\Sol(\mc{L}_i'))$ we have $p|_{G_i'}=0$. Hence, $\mb{I}(\Sol(\mc{L}_i'))\subseteq \Ideal{G_i'}$. This finishes the proof.
 \end{proof}
 
 Given the above lemma we prove $G'=\cup_{1\leq i\leq s} G'_i$ is a \GB\ for $\I(\mc{P}')=\mb{I}(\Sol(\mc{P}'))$ with respect to the \lex order~\eqref{lex-order}.
 \begin{lemma}
     \label{lem:GB-L'}
     The set of polynomials $G'=\cup_{1\leq i\leq s} G'_i$ forms a \GB\ for $\I(\mc{P}')=\mb{I}(\Sol(\mc{P}'))$ with respect to the \lex order $x_{1,1}\succ_\lex\dots\succ_\lex x_{1,k_1}\succ_\lex\dots\succ_\lex x_{s,1}\succ_\lex\dots\succ_\lex x_{s,k_s}\succ y_{1,1}\succ_\lex\dots\succ_\lex y_{1,r_1}\succ_\lex\dots\succ_\lex y_{s,1}\succ_\lex\dots\succ_\lex y_{s,r_s}$.
 \end{lemma}
 \begin{proof}
    Recall that 
    \begin{align*}
    \Sol(\mc{P}') &= \Sol(\mc{L}_1') \cap \dots \cap \Sol(\mc{L}_s')\\
    & = \Variety{\mb{I}(\Sol(\mc{L}_1'))} \cap \dots \cap \Variety{\mb{I}(\Sol(\mc{L}_s'))}\\
    &=  \Variety{\Ideal{G_1'}} \cap \dots \cap \Variety{\Ideal{G_s'}}\tag{by Lemma~\ref{lem:GB-P'}}\\
    &= \Variety{\Ideal{G_1'} + \dots + \Ideal{G_s'}} \tag{by Theorem~\ref{th:ideal_intersection}}
\end{align*}
This implies, 
\begin{align*}
   \I(\mc{P}')=\mb{I}(\Sol(\mc{P}')) =  \Ideal{G_1'}+\dots+\Ideal{G_s'}.
\end{align*}
Now by Lemma~\ref{lem:GB-P'} each $G_i'$ is a \GB. Moreover, observe that for all distinct $i$ and $j$ the ideals $\Ideal{G_i'}$ and $\Ideal{G_j'}$ do not share variables. Hence, the set of polynomials $G_1'\cup\dots\cup G_s'$ is indeed a \GB\ for $\mb{I}(\Sol(\mc{P}'))$, according to the \lex order.
 \end{proof}

\Cref{lem:equivalent} states that checking if a polynomial $p$ is in $\I(\mc{P})$ is equivalent to checking if the polynomial $p'$ is in $\I(\mc{P}')$. However, the substitution technique used in Lemma \ref{lem:equivalent} may result in a polynomial with exponentially many monomials and hence we only consider polynomials of bounded degree. Provided that $G'$ is a \GB\ for $\I(\mc{P}')$ with respect to the \lex order, see Lemma~\ref{lem:GB-L'}, we can test membership of any bounded degree polynomial in polynomial time, this gives us \Cref{thm:main-decision}. 

\section{Finding a proof and the substitution technique}
\label{sec:sub}

In~\cite{Bulatov20:ideal} we introduced a framework to bridge the gap between the decision and the search versions of the \IMP. Indeed, this framework gives a polynomial time algorithm to construct a truncated \GB\ provided that the search version of a variation of the \IMP\ is polynomial time solvable. This variation is called \xIMP\ and is defined as follows.

\begin{definition}[\xIMP]
    Given an ideal $\I\sse\Field[\vc x n]$ and a vector of $m$ polynomials $M=(g_1,\dots,g_m)$, the \xIMP~asks if there exist coefficients $\mathbf{c}=(\vc c m)\in \Field^m$ such that $\mathbf{c}M=\sum_{i=1}^m c_i g_i$ belongs to the ideal $\I$. In the search version of the problem the goal is to find coefficients $\bc$.
\end{definition}

    The \xIMP~associated with a (multi-sorted) constraint language $\Gamma$  over a set $\mc{D}$ is the problem \xIMP$(\Gamma)$ in which the input is a pair $(M,\cP)$ where $\cP$ is a
    $\CSP(\Gm)$ instance and $M$ is a vector of $m$ polynomials. The goal is to decide
    whether there are coefficients $\mathbf{c}=(\vc c m)\in \Field^m$ such that $\mathbf{c}M$ lies in the combinatorial ideal $\I(\cP)$. We use $\chi\IMP_d(\Gamma)$ to denote $\chi\IMP(\Gamma)$ when the vector $M$ contains polynomials of total degree at most $d$.
    
    \begin{theorem}[\cite{Bulatov20:ideal}]
    \label{thm:GB+xIMP}
        Let $\mc{H}$ be a class of ideals for which the search version of $\chi\IMP_d$ is polynomial time solvable. Then there exists a polynomial time algorithm that constructs a $d$-truncated \GB~of an ideal $\I\in \mc{H}$, $\I\subseteq\Field[\vc xn]$, in time $O(n^d)$.  
    \end{theorem}
    The above theorem suggests that, in order to prove the second part of Theorem~\ref{the:main-intro}, it is sufficient to show that \xIMP\ is polynomial time solvable for instances of \CSP\ arising from constraint languages that are closed under the affine operation of an Abelian group.

Similar to the $\IMP$, the $\xIMP$ behaves well with respect to pp-definitions. The following theorem is proved in \cite{Bulatov20:ideal} for one-sorted constraint languages, but it can be extended in a straightforward way to the multi-sorted case, in the same manner as in \Cref{thm:pp-interpret-multi-sorted}.

\begin{theorem}[\cite{Bulatov20:ideal}]\label{the:ximp-pp-defininition}
Let $\Gm,\Dl$ be constraint languages over the same collection of sets $\mc D$. If $\Gm$ pp-defines $\Dl$, then $\xIMP(\Dl)$ is polynomial time reducible to $\IMP(\Gm)$.
\end{theorem}

    It was shown that having a \GB\ yields a polynomial time algorithm for solving the search version of \xIMP\ (by using the \emph{division algorithm} and solving a system of linear equations). 
\begin{theorem}[\cite{Bulatov20:ideal}]
\label{thm:GB-XIMP}
    Let $\I$ be an ideal, and let $\{\vc g s\}$ be a \GB~for $\I$ with respect to some monomial ordering. Then the (search version of) \xIMP~ is polynomial time solvable.
 \end{theorem}
 
    Given the above theorem, to solve the \xIMP\ one might reduce the problem at hand to a problem for which a \GB\ can be constructed in a relatively simple way. This approach has been proven to be extremely useful in various cases studied in \cite{Bulatov20:ideal}. In that paper the reductions for \xIMP\ are proved in an ad hoc manner. However, the core idea in all of them is a substitution technique. Here we provide a unifying construction  based on \emph{substitution reductions} that covers all the useful cases so far.
    
\subsection{Reduction by substitution}
\label{sec:redu-by-sub}

We call a subproblem of the IMP or $\xIMP$ \emph{CSP-based} if its instances are of the form $(f,\cP)$ or $(M,\cP)$, where $\cP$ is a CSP instance over a fixed set $D$. Let $\mathcal X,\mathcal Y$ be restricted CSP-based subproblems of the \xIMP. The problems $\mathcal X,\mathcal Y$ can be defined by various kinds of restrictions, for example, as $\xIMP(\Gm),\xIMP(\Dl)$, but not necessarily. Let the domain of $\cX$ be $D$ and the domain of $\cY$ be $E$. Let also $\vc\mu k$ be a collection of surjective functions $\mu_i:E^{\ell_i}\to D$, $i\in[k]$. Each mapping $\mu_i$ can be interpolated by a polynomial $h_i$. We call the collection $\{\vc hk\}$ a \emph{substitution collection} and call $\ell_i$ the substitution arity of $h_i$. 

We define substitution reductions for the $\xIMP$. For the IMP it is quite similar. The problem $\mathcal X$ is said to be \emph{substitution reducible} to $\mathcal Y$ if there exists a substitution collection $\{\vc hk\}$ and a polynomial time algorithm $A$ such that for every instance $(M,\cP)$ of $\cX$ the instance constructed as follows belongs to $\cY$.
\begin{itemize}
    \item[(1)] 
    Let $X$ be the set of variables of $(M,\cP)$. For every $x\in X$ the algorithm $A$ selects a polynomial $h_{i_x}$ and a set of variables $Y_x$ such that 
    \begin{itemize}
        \item[(a)]
        $|Y_x|=\ell_{i_x}$;
        \item[(b)]
        for any $x,y\in X$ either $Y_x=Y_y$ or $Y_x\cap Y_y=\emptyset$;
        \item[(c)]
        if $\vc xr\in X$ are such that $Y_{x_1}=\dots=Y_{x_r}=\{\vc y{\ell_j}\}$ then for any solution $\vf$ of $\cP$ there are values $\vc a{\ell_j}\in E$ such that $\vf(x_i)=h_{i_{x_i}}(\vc a{\ell_j})$.
    \end{itemize}
    \item[(2)]
    If $M=(\vc g m)$ then $M'=(\vc{g'} m)$, where for $g_i(\vc xt)$
    \[
        g'_i=g_i(h_{i_{x_1}}(Y_{x_1})\zd h_{i_{x_t}}(Y_{x_t})).
    \]
    \item[(3)]
    Let $Y=\bigcup_{x\in X}Y_x$. The instance $\cP'$ is given by $(Y,E,\cC')$, where for every constraint $\ang{\bs,R}$, $\bs=(\vc xt)$, $\cP'$ contains the constraint $\ang{\bs',R'}$ such that
    \begin{itemize}
        \item[--]
        $\bs'=(x_{1,1}\zd x_{1,\ell_{x_1}},x_{2,1}\zd x_{t,\ell_{x_t}})$, where $Y_{x_j}=\{x_{j,1}\zd x_{j,\ell_j}\}$;
        \item[--]
        $R'$ is an $\ell$-ary relation, $\ell=\ell_{x_1}+\zd+\ell_{x_t}$, such that $(a_{1,1}\zd a_{1,\ell_{x_1}},a_{2,1}\zd a_{t,\ell_{x_t}})\in R'$ if and only if $(h_{i_{x_1}}(a_{1,1}\zd a_{1,\ell_{x_1}})\zd h_{i_{x_t}}(a_{t,1}\zd a_{t,\ell_{x_t}}))\in R$.
    \end{itemize}
\end{itemize}

\begin{example}\label{exa:substitution}
(1) Let $D=\{0,1,2\}$, $E=\{0,1\}$, $\mathcal X,\mathcal Y$ are of the form $\xIMP(\Gm),\xIMP(\Dl)$, respectively, where $\Gm=\{R\}$ is a constraint language on $D$, and $\Dl=\{R'\}$ is a constraint language on $E$ with
\[
R=\left(\begin{array}{ccccc}
0&1&2&0&1\\
0&1&2&2&2 
\end{array}\right),\qquad
R'=\left(\begin{array}{ccccc}
0&1&1&0&1\\
1&0&1&1&0\\
0&1&1&1&1\\
1&0&1&1&1
\end{array}\right).
\]
Note that $R$ is partial order on $D$. Let $\mu_1:E^2\to D$ given by $\mu_1(0,0)=\mu_1(0,1)=0$, $\mu_1(1,0)=1$, $\mu_1(1,1)=2$. Then the polynomial $h_1(x,y)=xy+x$ interpolates $\mu_1$.

Consider the instance of $\mathcal X$: $X=\{x,y,z\}$, $\mathcal C=\{C_1=\ang{(x,y),R}, C_2=\ang{(y,z),R}$, $f=xy+yz$. In order to perform substitution we set $Y_x=\{x_1,x_2\},Y_y=\{y_1,y_2\},Y_z=\{z_1,z_2\}$. Then $Y=\{x_1,x_2,y_1,y_2,z_1,z_2\}$ and the constraints are converted as follows: $C_1$ is converted into $\ang{(x_1,x_2,y_1,y_2),R'}$ and $C_2$ is converted into $\ang{(y_1,y_2,z_1,z_2),R'}$. Finally, $f$ becomes
\[
f(h_1(x_1,x_2),h_1(y_1,y_2),h_1(z_1,z_2))=(x_1x_2+x_1)(y_1y_2+y_1)+(y_1y_2+y_1)(z_1z_2+z_1).
\]
It is straightforward to verify that conditions (1)--(3) hold for this transformation showing that $\mathcal X$ is substitution reducible to $\mathcal Y$.
\end{example}

\begin{lemma}
\label{lem:reduct-sub}
     Let $\mathcal X,\mathcal Y$ be restricted CSP-based subproblems of the $\xIMP_d$ and $\xIMP_{r d}$, respectively. If $\mc X$ is substitution reducible to $\mc{Y}$ with a substitution collection $\{\vc hk\}$ of substitution arities $\{\vc\ell k\}$, and $r\ge\ell_i$ for each $i\in[k]$, then there is a polynomial time reduction from $\mc{X}$ to $\mc{Y}$. 
\end{lemma}

\begin{proof}
    Let $(M,\mc{P})$ with $M=(\vc g m)$ be an instance of $\mathcal X$. Moreover, suppose polynomials in $M$ have total degree at most $d$. By the definition of substitution reducibility, in polynomial time, we construct an instance $(M',\mc{P}')$ with $M'=(\vc {g'}m)$ from $\mathcal Y$ that satisfies the conditions in the definition. Note that each polynomial $h_i$ in the substitution collection has degree at most $\ell_i|D|$, therefore, each $g'_i$ in $M'$ has a bounded degree. We now prove $(M,\mc{P})$ is a \textbf{yes} instance if and only if $(M',\mc{P}')$ is a \textbf{yes} instance.
    
    Recall that  $Y=\bigcup_{x\in X}Y_x$ and set $X=\{\vc xn\}$, $Y=\{\vc y{n'}\}$. Let $\I(\mc{P})\subseteq \Field[X]$ and $\I(\mc{P}')\subseteq \Field[Y]$ be the ideals corresponding to $\mc{P}$ and $\mc{P}'$, respectively. Consider an arbitrary $\mathbf{c}\in \Field^m$ and set $f\in \Field[X]$ to be
    \begin{align*}
        f(\vc xn) &= \sum_{i=1}^m c_i g_i(\vc xn)\\
        &=\mathbf{c}M.
    \end{align*}
    Now define the polynomial $f'\in\Field[Y]$ to be
    \begin{align*}
        f' (\vc ym)&= \sum_{i=1}^m c_i g_i(h_{i_{x_1}}(Y_{x_1}),\dots,h_{i_{x_n}}(Y_{x_n}))\\
        &= \sum_{i=1}^m c_i g'_i(Y)\\
        &=\mathbf{c}M'.
    \end{align*}
    
    In what follows, we prove that we can construct a satisfying assignment for $\mc{P}'$ from a satisfying assignment for $\mc{P}$, and vice versa. Consider a satisfying assignment $\psi$ for the instance $\mc{P}'$. We find a mapping $\vf:X\to D$ and show it is a satisfying assignment for $\mc{P}$. Define $\vf$ as follows. For any $x\in X$ with $Y_x=\{y_{x,1},\dots,y_{x,\ell_{i_x}}\}$ let 
    \begin{align*}
        \vf(x)= h_{i_x}(\psi(y_{x,1}),\dots,\psi(y_{x,\ell_{i_x}})).
    \end{align*}
    Consider a constraint from $\mc{P}$, say $\ang{\bs,R}$ with $\bs=(x_1,\dots,x_t)$. By definition, there exists a constraint $\ang{\bs',R'}$ such that 
    \begin{itemize}
        \item[$-$]
        $\bs'=(x_{1,1}\zd x_{1,\ell_{x_1}},x_{2,1}\zd x_{t,\ell_{x_t}})$, where $Y_{x_j}=\{x_{j,1}\zd x_{j,\ell_j}\}\subseteq Y$;
        \item[$-$]
        $R'$ is an $\ell$-ary relation, $\ell=\ell_{x_1}+\dots+\ell_{x_t}$, such that $(a_{1,1}\zd a_{1,\ell_{x_1}},a_{2,1}\zd a_{t,\ell_{x_t}})\in R'$ if and only if $(h_{i_{x_1}}(a_{1,1}\zd a_{1,\ell_{x_1}})\zd h_{i_{x_t}}(a_{t,1}\zd a_{t,\ell_{x_t}}))\in R$.
    \end{itemize}
    Now since $\psi$ is a satisfying assignment of $\mc{P}'$ then 
    \[
        (\psi(x_{1,1})\zd \psi(x_{1,\ell_{x_1}}),\psi(x_{2,1})\zd \psi(x_{t,\ell_{x_t}}))\in R'
    \] 
    if and only if 
    \[
        \left(h_{i_{x_1}}(\psi(x_{1,1})\zd \psi(x_{1,\ell_{x_1}}))\zd h_{i_{x_t}}(\psi(x_{t,1})\zd \psi(x_{t,\ell_{x_t}}))\right)\in R.
    \]
    Hence, since $\psi$ is a satisfying assignment for $\mc{P}'$ then $\vf$ is a satisfying assignment for $\mc{P}$.
    
    Conversely, consider a satisfying assignment $\vf$ for the instance $\mc{P}$. We find a mapping $\psi:Y\to E$ and show it is a satisfying assignment for $\mc{P}'$. Define $\psi$ as follows. Let $x_1,\dots,x_r$ be all the variables such that $y\in Y_{x_i}$, $i\in [r]$. According to the definition, item 1(b), we must have $Y_{x_1}=\dots=Y_{x_r}=\{y_1,y_2,\dots,y_{\ell_j}\}$. Without loss of generality, suppose $y=y_1$. Now, according to item 1(b) of the definition, since $\vf$ is a solution of $\mc{P}$ there are values $a_1,\dots,a_{\ell_{j}}\in E$ such that $\vf(x_i)=h_{i_{x_i}}(\vc a{\ell_j})$ for all $ i\in [r]$. Hence, in this case we set $\psi(y)=a$.

    
    Now we show that $\psi$ is a satisfying assignment for $\mc{P}'$. Consider a constraint from $\mc{P}'$, let's say $\ang{\bs',R'}$. By item (3) in the definition, there exists a constraint $\ang{\bs,R}$ with $\bs=(x_1,\dots,x_t)$ in $\mc{P}$ that gives rise to $\ang{\bs',R'}$ such that 
    \begin{itemize}
        \item[$-$]
        $\bs'=(x_{1,1}\zd x_{1,\ell_{x_1}},x_{2,1}\zd x_{t,\ell_{x_t}})$, where $Y_{x_j}=\{x_{j,1}\zd x_{j,\ell_j}\}\subseteq Y$;
        \item[$-$]
        $R'$ is an $\ell$-ary relation, $\ell=\ell_{x_1}+\zd+\ell_{x_t}$, such that $(a_{1,1}\zd a_{1,\ell_{x_1}},a_{2,1}\zd a_{t,\ell_{x_t}})\in R'$ if and only if $(h_{i_{x_1}}(a_{1,1}\zd a_{1,\ell_{x_1}})\zd h_{i_{x_t}}(a_{t,1}\zd a_{t,\ell_{x_t}}))\in R$.
    \end{itemize}
    Now since $\vf$ is a satisfying assignment of $\mc{P}$ then 
    \begin{align*}
        &(\vf(x_1)\zd \vf(x_{n}))  \in R \\
        \implies &\left(h_{i_{x_1}}(\psi(x_{1,1})\zd \psi(x_{1,\ell_{x_1}}))\zd h_{i_{x_t}}(\psi(x_{t,1})\zd \psi(x_{t,\ell_{x_t}}))\right) \in R
    \end{align*}
    if and only if 
    \[
        (\psi(x_{1,1})\zd \psi(x_{1,\ell_{x_1}}),\psi(x_{2,1})\zd \psi(x_{t,\ell_{x_t}}))\in R'
    \] 
     Hence, since $\vf$ is a satisfying assignment for $\mc{P}$ then $\psi$ is a satisfying assignment for $\mc{P}'$.
    
     It remains to show $f'(\psi)=0$ if and only if $f(\vf)=0$. This is the case because
     \begin{align*}
        f'(\psi) &=  \sum_{i=1}^m c_i g'_i(\psi)\\
        &= \sum_{i=1}^m c_i g_i\left(h_{i_{x_1}}(\psi(Y_{x_1})),\dots,h_{i_{x_n}}(\psi(Y_{x_n}))\right)\\
        &= \sum_{i=1}^m c_i g_i(\vf(x_1),\dots,\vf(x_n))\\
        &= f(\vf)
    \end{align*}
    Therefore we have proved
    \[
        \exists \mathbf{c}\in \Field^m \text{~with~} \mathbf{c}M\in \I(\mc{P}) \iff \exists \mathbf{c}'\in \Field^{m} \text{~with~} \mathbf{c}'M'\in \I(\mc{P}').
    \]
    This finishes the proof.
\end{proof}

Theorem~\ref{thm:GB-XIMP} and the above lemma provide a powerful tool for solving the \xIMP. That is, if $\mc{X}$ is substitution reducible to $\mc{Y}$ and furthermore $\mc{Y}$ is such that it admits a polynomial time algorithm to construct a \GB, then instances of $\mc{X}$ are solvable in polynomial time too. More formally,

\begin{theorem}
\label{thm:sub+GB}
    Let $r \geq 1$ and let $\mc{X},\mc{Y}$ be restricted \CSP-based subproblems of the $\xIMP_d$ and $\xIMP_{rd}$, respectively, such that $\mc{X}$ is substitution reducible to $\mc Y$ with a substitution collection $\{\vc hk\}$ of substitution arities $\{\vc\ell k\}$ and $r\ge\ell_i$ for $i\in[k]$. Suppose there exists a polynomial time algorithm that for any instance $(M',\mc P')$ of $\mc Y$ constructs a (truncated) \GB, then
    \begin{enumerate}
        \item every instance $(M,\mc P)$ of the search version of $\mc X$ is polynomial time solvable;
        \item there exists a polynomial time algorithm that for any instance $(M,\mc P)$ of $\mc X$ constructs a $d$-truncated \GB\ for $\I(\mc P)$. 
    \end{enumerate}
\end{theorem}

\begin{proof}
    Suppose $M$ contains $m$ polynomials $\vc gm$. From instance $(M,\mc P)$ of $\mc X$ we construct an instance $(M',\mc P')$ of $\mc Y$ as explained above. Now by Lemma~\ref{lem:reduct-sub} these two instances are equivalent. 
    
    The objective is to find $\mathbf{c}\in\Field^m$ such that $f=\mb{c}M\in\I(\mc{P})$, if one exists. After carrying out the construction the coefficients in all the polynomials $\vc {g'}m\in M'$ are linear combination of elements of $\mb{c}$. Hence, we have polynomial $f'=\mb{c}'M'$ where each entry of $\mb{c}'$ is a linear combination of elements of $\mb{c}$. By our assumption, if we can construct a \GB\ for $\I(\mc{P}')$ then we can check in polynomial time if such $\mb{c}'$ exists. If no such $\mb{c}'$ exists then $(M,\mc P)$ is a \textbf{no} instance, else we can solve a system of linear equations over the elements of $\mb{c}$ and find a solution $\mb{c}$. 
    
    Finally, for the second part of the theorem since $(M,\mc P)$ is polynomial time solvable, by Theorem~\ref{thm:GB+xIMP}, we can construct a $d$-truncated \GB\ for $\I(\mc{P})$ in polynomial time.
\end{proof}

\subsection{Applications of reduction by substitution}
In this section we demonstrate that the notion of reduction by substitution introduced in the previous section is applicable to various cases, in particular the case of \CSP s over constraint languages closed under the affine operation of a finite Abelian group. We start off with the case of pp-interpretation.

\begin{lemma}
\label{lem:sub+pp-int}
    Let $\Gm$ and $\Dl$ be multi-sorted constraint languages over finite collection of sets $\mc D=\{D_t\mid t\in T\}$, $\mc E=\{E_s\mid s\in S\}$, respectively. Let $\mc X$ be a problem of the of the form $\xIMP_d(\Dl)$ for some $d$, and suppose $\Gm$ pp-interprets $\Dl$. Then there is a constraint language $\Gm'$ on $\mc E$ pp-definable in $\Gm$ such that for the CSP-based problem $\mc Y$ defined as $\xIMP(\Gm')$, the problem $\mc X$ is substitution reducible to $\mc Y$.
\end{lemma}

\begin{proof}
Let $F_s,\pi_s$, $s\in S$, be the sets and mappings given in a pp-interpretation of $\Dl$ by $\Gm$, see Definition~\ref{pp-interpret-multi-sorted}. Set
\[
\Gm'=\{\pi^{-1}(Q)\mid Q\in\Dl\}.
\]
By the requirements of pp-interpretations, $\Gm'$ is pp-definable in $\Gm$.

Now, let $(M, \mc{P})$ be an instance of $\mc X$. We provide an algorithm that constructs an instance $(M',\mc{P}')$ of $\mc Y$ in such a way that it satisfies the conditions for reduction by substitution.
    
Define the substitution collection $\{h_s\mid s\in S\}$ to be the set of polynomials where each $h_s$ interpolates the onto mapping $\pi_s : F_s \to E_s$. For every constraint $\ang{\mb{v},R}$ in $\mc{P}$ with $\mb{v}=(\vc xt)$, $\mc{P}'$ contains the constraint $\ang{\mb{v}',R'}$ with
\begin{itemize}
\item [$-$] $\mb{v}'=(x_{1,1},\dots,x_{1,\ell_{s_1}},\dots,x_{t,1},\dots,x_{t,\ell_{s_t}})$, and
\item [$-$] $R'$ is such that
\[
\pi^{-1}(R)(x_{1,1},\ldots , x_{1,\ell_{s_1}},x_{2,1},\ldots,x_{2,\ell_{s_2}},\ldots,x_{t,1},\ldots,x_{t,\ell_{s_t}})\qquad \text{is true}
\]
if and only if
\[
R(h_{s_1}(x_{1,1},\ldots , x_{1,\ell_{s_1}}),\dots,h_{s_k}(x_{k,1},\ldots,x_{k,\ell_{s_t}})) \qquad\text{is true}.
\]
\end{itemize}
By the definition of $\Gm'$, $R'\in\Gm'$. Now for each $x_i$ with $\delta_\Dl(x_i)=s$ we have $Y_{x_i}=\{x_{i,1},\ldots , x_{i,\ell_{s}}\}$. Note that for every distinct $x_i$ and $x_j$ we have $Y_{x_i}\cap Y_{x_j}=\emptyset$. This satisfies conditions 1(b),(c). Moreover, according to pp-interpretability and the way $(M',\mc{P}')$ is constructed condition (3) is also satisfied.
\end{proof}

From Lemma~\ref{lem:sub+pp-int} and Theorems~\ref{the:ximp-pp-defininition},~\ref{thm:sub+GB} we obtain the following corollary.

\begin{corollary}
    Let $\Gm$ and $\Dl$ be multi-sorted constraint languages over finite collection of sets $\mc D=\{D_t\mid t\in T\}$, $\mc E=\{E_s\mid s\in S\}$, respectively. Suppose $\Gm$ pp-interprets $\Dl$ and there exists a polynomial time algorithm that  for any instance $(M',\mc P')$ of $\xIMP_{O(d)}(\Gm)$ constructs a (truncated) \GB, then 
    
    \begin{enumerate}
        \item every instance $(M,\mc P)$ of $\xIMP_d(\Dl)$ is polynomial time solvable;
        \item there exists a polynomial time algorithm that for any $d\in\zN$ and any instance $\cP$ of $\CSP(\Dl)$ constructs a $d$-truncated \GB\ for $\I(\mc P)$. 
    \end{enumerate}
\end{corollary}

\subsubsection{Reduction by substitution for languages over Abelian groups}
In this section we prove that our reductions for constraint languages over finite Abelian groups is an example of reduction by substitution. Lemma~\ref{lem:sub+pp-int} states that reductions under pp-interpretability can be seen as reduction by substitution. This means the part of our reduction where we transform an instance of $\CSP$ over an Abelian group to an instance of $\CSP$ over $\zZ_{p_1^{m_1}}\zd\zZ_{p_s^{m_s}}$ can be seen as a reduction by substitution. We will show that the reduction to roots of unity is also a reduction by substitution.

Let $(M,\mc{P})$ be such that $M=(\vc gm)$ is a vector of polynomials of length $m$ where each $g_i\in M$ is from $\Complex[x_{1,1},\dots,x_{1,k_1},\dots, x_{s,1},\dots,x_{s,k_s}]$ and $\mc{P}$  an instance of $\CSP(\Gm)$. Here $\Gm$ is a constraint language invariant under the affine operation of $\zZ_{p_1^{m_1}}\zd\zZ_{p_s^{m_s}}$. Moreover, $\mc{P}$ can be represented as a collection of systems of linear equations $\mc{L}_1,\dots,\mc{L}_s$ where
\begin{itemize}
    \item[1.] each $\mc{L}_i$ is a system of linear equations over $\zZ_{p_i^{m_i}}$ with variables $X(\mc L_i)\cup Y(\mc Y_i)$, $X(\mc{L}_i)=\{x_{i,1},\dots,x_{i,k_i}\}$, $Y(\mc{L}_i)=\{y_{i,1},\dots,y_{i,r_i}\}$; 
    \item[2.] each $\mc{L}_i$ is of the following form
        \begin{align*}
            (\mathbb{1}_{k_i\times k_i}~~ B_i)(x_{i,1},\dots,x_{i,k_i},y_{i,1},\dots,y_{i,r_i},1)^T=\mathbf{0};
        \end{align*}
    \item[3.] $X(\mc{L}_i)\cap X(\mc{L}_j)=\emptyset$, $Y(\mc{L}_i)\cap Y(\mc{L}_j)=\emptyset$, for all $1\le i,j\leq s$ and $i\neq j$;
    \item[4.] an assignment $\vf$ to variables from $X$ is a solution of $\cP$ if and only if for every $i\in[s]$ there are values of variables from $Y(\mc L_i)$ that together with $\vf_{|X(\mc L_i)}$ satisfy $\mc L_i$.
\end{itemize}
Now for each $i\in[s]$ let $h_i$ be a polynomial that interpolates the mapping 
\[
    (0,\om_i^0),(1,\om_i),\dots,(p_i^{m_i}-1,\om_i^{(p_i^{m_i}-1)})
\]
where $\om_i$ is a primitive $p_i^{m_i}$-th root of unity. The substitution collection consists of all $h_i$, $i\in [s]$. For every variable $x_{i,j}, y_{i,j}$ we set $Y_{i,j}=\{x_{i,j}\}$ and $Y'_{i,j}=\{y_{i,j}\}$ satisfying condition 1(b). Then, for each variable $x_{i,j},y_{i,j}$ we choose $h_{i_{x_{i,j}}}=h_i$, $h_{i_{y_{i,j}}}=h_i$. Thus, condition 1(c) is satisfied.

Now for every constraint in $\mc{P}$ which is of the form 
\[
    x_{i,t} + \alpha_{t,1}~y_{i,1}+\dots + \alpha_{t,r_i}~y_{i,r_i}+\alpha_t  =  0  \pmod {p_i^{m_i}}
\]
we add the following constraint in $\mc{P}'$
\[
     x_{i,t}-\om_i^{\alpha_t}\cdot\left(y_{i,1}^{\alpha_{t,1}}\cdot\ldots \cdot y_{i,r_i}^{\alpha_{t,r_i}}\right) = 0.
\]
Such construction of $\mc{P}'$ guarantees that conditions in (3) hold. Now it is immediate that our transformation of the problem to an equivalent problem over roots of unities is indeed a reduction by substitution. This together with the fact that there exists an algorithm to construct a \GB\ for the equivalent problem over roots of unities, see Lemma~\ref{lem:GB-L'}, give us the following theorem. 
\begin{theorem}[\Cref{the:main-intro} paraphrased]
    Let $\zA$ be a finite Abelian group. Then $\IMP_d(\Dl)$ is polynomial time solvable for any finite constraint language $\Dl$ which is invariant under the affine operation of $\zA$. In fact, there exists a polynomial time algorithm that for any $d\in\zN$ and any instance $\cP$ of $\CSP(\Dl)$ constructs a $d$-truncated \GB\ for $\I(\mc P)$.
\end{theorem}
\begin{proof}
    The discussion above tells us that $\xIMP_d(\Dl)$ is substitution reducible to a class of \xIMP, say $\mc Y$, where for every instance $(M',\mc{P}')$ of $\mc Y$ we can construct a \GB. Hence, by \Cref{thm:sub+GB}, every instance $(M,\mc P)$ of $\chi\IMP_d(\Dl)$ is polynomial time solvable. Moreover, by item 2 of \Cref{thm:sub+GB}, we can construct a $d$-truncated \GB\ for $\I(\mc{P})$ thus $\IMP_d(\Dl)$ is polynomial time solvable.
\end{proof}

\addcontentsline{toc}{section}{References}
\bibliographystyle{plain}
\bibliography{Reference}

\end{document}